\documentclass[11pt]{amsart}
\usepackage{amsmath}
\usepackage{amssymb}
\usepackage{tikz}
\usepackage{color}
\usepackage{xcolor}

\usepackage{geometry}
 \theoremstyle{definition}

 \newtheorem{theorem}{Theorem}[section]%
 
 \newtheorem{lemma}[theorem]{Lemma}
 
 \newtheorem{proposition}[theorem]{Proposition}
 
  \numberwithin{equation}{section}

\theoremstyle{remark}

\newtheorem{remark}{Remark}

\newcommand\sgn{\mathrm{sgn}}

\newcommand\calD{\mathcal{D}}
\newcommand\C{\mathbb{C}}
\newcommand\N{\mathbb{N}}
\newcommand\R{\mathbb{R}}
\newcommand{\bra}[1]{\left \langle #1 \right \rvert}
\newcommand{\ket}[1]{\left \lvert #1 \right \rangle}
\def\vecg{\vec{\mathfrak{g}}_1}
\def\vecgg{\vec{\mathfrak{g}}_2}
\def\vecG{\vec{\mathfrak{G}}_1}
\def\vecGG{\vec{\mathfrak{G}}_2}

\DeclareMathOperator{\trace}{Tr}

\title[Entanglement entropy in a fermionic chain]{Entanglement entropy of two disjoint intervals separated by one spin in a chain of free fermion}

\author{L. Brightmore}
\address{L. Brightmore, School of Mathematics, University of Bristol, Fry Building, Woodland Road, Bristol BS8 1UG, UK}

\author{G.P. Geh\'er}
\address{G.P. Geh\'er, Department of Mathematics and Statistics, University of Reading, Whiteknights, P.O.
Box 220, Reading RG6 6AX, United Kingdom}
\email{G.P.Geher@reading.ac.uk or gehergyuri@gmail.com \newline
\hspace{1.6cm} http://www.math.u-szeged.hu/\~{}gehergy}

\author{A.R. Its}
\address{A.R. Its, Department of Mathematical Sciences, Indiana University-Purdue University Indianapolis, 402 N. Blackford St., Indianapolis, IN 46202-3267, United States of America}
\email{aits@iupui.edu}

\author{V.E. Korepin}
\address{V.E. Korepin, C.N. Yang Institute for Theoretical Physics, State University of New York at Stony Brook, Stony Brook, NY 11794-3840, USA}
\email{korepin@gmail.com}

\author{F. Mezzadri}
\address{F. Mezzadri, School of Mathematics, University of Bristol, Fry Building, Woodland Road, Bristol BS8 1UG, UK}
\email{F.Mezzadri@bristol.ac.uk}

\author{M.Y. Mo}
\address{M.Y. Mo, School of Mathematics, University of Bristol, Fry Building, Woodland Road, Bristol BS8 1UG, UK}

\author{J.A. Virtanen}
\address{J.A. Virtanen, Department of Mathematics and Statistics, University of Reading, Whiteknights, P.O.
Box 220, Reading RG6 6AX, United Kingdom}
\email{j.a.virtanen@reading.ac.uk}

\begin{document}



\thanks{Geh\'er was supported by the Leverhulme Trust Early Career Fellowship (ECF-2018-125), and by the Hungarian National Research, Development and Innovation Office (Grant no. K115383).
Its was supported by the NSF grant DMS-1700261.
Virtanen was supported in part by EPSRC grants EP/M024784/1 and EP/T008636/1.
Geh\'er and Virtanen also thank the American Institute of Mathematics and the SQuaRE program for their support.}

\begin{abstract}
We calculate the entanglement entropy of a non-contiguous subsystem of a chain of free fermions. The starting point is a formula suggested by Jin and Korepin, \texttt{arXiv:1104.1004}, for the reduced density of states of two disjoint intervals with lattice sites $P=\{1,2,\dots,m\}\cup\{2m+1,2m+2,\dots, 3m\}$, which applies to this model. As a first step in the asymptotic  analysis of this system, we consider its simplification to two disjoint intervals separated just by one site, and  we rigorously calculate the mutual information between these two blocks and the rest of the chain. In order to compute the entropy  we need to study the asymptotic behaviour of an inverse Toeplitz matrix with Fisher-Hartwig symbol using the the Riemann--Hilbert method.
\end{abstract}

\maketitle
\tableofcontents

\section{Introduction}

Quantum systems that are spatially separated can share information
that cannot be accounted for by the relativistic laws of classical
physics.  This fundamental property of quantum mechanics, which plays
a crucial role in quantum information, is known as \emph{entanglement}
and its measurement is still largely an open
problem. 
There is not a unique way of quantifying entanglement; however, in
bipartite systems one of the most popular and successful measure of
entanglement is the von Neumann entropy~\cite{BBPS96}. 

Suppose that the system is in a pure state $\ket{\psi}$.  The density
matrix is simply the projection operator $\rho_{PQ} =
\ket{\psi}\mspace{-6.0 mu} \bra{\psi}$, where $P$ and $Q$ refer to the
two parts and the Hilbert space $\mathcal{H} = \mathcal{H}_P \otimes
\mathcal{H}_Q$.  The von Neumann entropy is defined as
\begin{equation}
  \label{v-n entropy}
  S\left(\rho_P\right) = S\left(\rho_Q\right) = - \trace \left(\rho_P
    \log \rho_P\right) = - \trace \left(\rho_Q
    \log \rho_Q\right),  
\end{equation}
where
\begin{equation}
  \label{eq:rhoA+rhoB}
  \rho_P = \trace_Q \rho_{PQ},  \quad \rho_Q = \trace_P
  \rho_{PQ} 
\end{equation}
and $\trace_P$ and $\trace_Q$ denote the partial traces over the
degrees of freedom of $P$ and $Q$, respectively.

In this paper we study the entropy of a two-block subsystem is a chain of free fermions.  More precisely, we consider 
the chain 
\begin{equation}\label{Fchain}
H_F = -\sum_{j=1}^{N}b_j^\dagger b_{j+1} + b_jb_{j+1}^\dagger,
\end{equation}
where the Fermi operators $b_j$ are defined by the anticommutation relations
\begin{equation}
\{b_j,b_k\} = 0 \quad \text{and} \quad \{b_j,b_k^\dagger\}= \delta_{jk}. 
\end{equation}
The starting point for this analysis is an integral representation for the von Neumann entropy of
the subsystem $P$ of fermions on lattice sites
\begin{equation}
\label{Ammm}
P=\{1,2,\dots , m\}\cup \{2m+1,2m+2,\dots, 3m\}.
\end{equation}
This was derived by B.-Q. Jin and the fourth co-author  in
~\cite{JK11}, and followed on from the success of this approach
to calculating the entropy of a contiguous block of spins in the XX model~\cite{JK04}.
Our goal is to compute the entanglement entropy between the subsystem~\eqref{Ammm} and the 
rest of the chain in the limit as $m \to \infty$. 

Over the past two decades the entanglement of bipartite
systems have been extensively studied in one-dimensional quantum critical
systems, in particular quantum spin chains.  Consider a spin chain with $N$ spins; at zero
temperature the Hamiltonian is in the ground state and in the
thermodynamic limit $N \to \infty$ it undergoes a phase transition
for some critical value of a parameter, \emph{e.g.}  the magnetic
field.  This quantum phase transition is characterized by an infinite
spin-spin correlation length.  Several papers have
addressed the problem of computing the entanglement of the first $L$ consecutive 
spins and the rest of the chain in various 
contexts~\cite{OAFF02,ON02,AOPFP04,VLRK03,JK04,KM04,Kor04,KM05,IJK05,IMM08}.
It is well known that the entanglement entropy grows as
\begin{equation}
  \label{eq:oneinterval}
  S\left(\rho_P\right) \propto \log L, \qquad L \to \infty.
\end{equation}
Recently, there has been considerable interest in computing
$S\left(\rho_P\right)$ in quantum spin chains when $P$ is made of disjoint regions of space. Up to now
this problem has received attention within the framework  Conformal Field Theory 
(CFT)~\cite{CCT09,CCT11,Car13,ATC10,FC10,FPS09}.  One-dimensional quantum critical systems
can be described in terms of a massless CFT. 
More general holographic descriptions are given in \cite{Tak} and  ~\cite{MVS11}.
 When $P$ is one
interval, then the coefficient of the logarithm
in~\eqref{eq:oneinterval} is proportional to the central charge $c$,
which is a characteristic of the theory~\cite{CaCa04}.     If the theory is bosonic,
\textit{i.e.} if $c$ is an integer, then in the two-interval case the von
Neumann entropy  depends on the compactification radius of the bosonic
field~\cite{FPS09}.
In the papers~\cite{CCT09,CCT11}  the moments of the density matrix were obtained for two-intervals as ratios of Jacobi theta-functions.  Unfortunately, they could not compute the analytic continuation of their formulae in terms of the exponent of the moments, which would have led them to an expression for the von Neumann entropy, except in the asymptotic limit of small intervals~\cite{CCT11}.   

A well established approach to solve quantum spin chains that goes back to Lieb \emph{et al.}~\cite{LES61} 
is to map the spin operators into Fermi operators using the Jordan-Wigner transformation.  For example, the 
XX chain 
\begin{equation}\label{Hxxb}
H_{XX} = -\sum_{j=1}^{N}\sigma_j^{x}\sigma^{x}_{j+1} + \sigma_j^{y}\sigma^{y}_{j+1}
\end{equation}
is mapped into~\eqref{Fchain}. 
This approach works well when computing the von Neuman and Reny entropies of a single contiguous interval, as 
the entropy of the first $L$ spins coincides with that of the first $L$ fermions in~\eqref{Fchain}. However, in the
case of disjoint intervals in a spin chain  there is the extra complication due to the fact that in the fermionic space the operators 
between blocks contribute to the entropy, because the Jordan-Wigner transformation is not local.
This problem was tackled using CFT by Fagotti and Calabrese~\cite{FC10}. In order to avoid this technicality, our starting 
point is the fermionic chain~\eqref{Fchain}.  In the model~\eqref{Fchain}, 
 the Fermi operators in between blocks do not appear in the computation 
of the reduced density of states; therefore, the approach adopted in~\cite{JK11} applies.  This simplification allows a rigorous computation of the asymptotic behaviour of the entanglement entropy as $m \to \infty$ while at the same time preserving the physical 
phenomenon that we want to study.  This idea is not new and was adopted by Ares \emph{et al.}~\cite{AEF14}, who performed 
a numerical study and conjectured a 
formula of the entropy of several disjoint blocks in a chain of Fermi operators.  In fact, our main result --- formula (\ref{eq:mut-inf})
--- seems to be consistent with Ares-Esteve-Falceto conjecture. We hope to address this issue in all detail in the forthcoming publication.


One of the main features of this representation of the von Neumann entropy derived in~\cite{JK11} for the
two-blocks~\eqref{Ammm} is that the
computation of the entanglement \textcolor{black}{reduces to an integral involving the determinant of a block-matrix, whose two block-diagonal entries are Toeplitz determinants, see formulae (41), (48)--(51) in \cite{JK11}, or \eqref{eq:A1}--\eqref{eq:entr-int} below.}
This calculation would be the ultimate
goal, but at the moment it is out of our reach --- Remark 1 in Section 3. 
In this paper, instead,
we consider a simplified example of a subsystem consisting of two intervals
separated by just one lattice site.  The  asymptotic analysis of this model is already
much more difficult  than that of a single block Hamiltonian. Indeed, we not only have to evaluate the asymptotics of the Toeplitz determinant itself, but we also need to extract detailed information on the asymptotic behaviour {of the inverse} Toeplitz matrix. 

It should also be noticed that,  besides its intrinsic interest
as a physical problem, the study of the asymptotics of Toeplitz
determinants has a long history going back to Szeg\H{o}~\cite{Sze15,Sze52} as such matrices are ubiquitous in
mathematics and physics.  Indeed, starting from the seminal works of Szeg\H{o}, Kaufman and Onsager, the Toeplitz
determinants have been playing a very important role in many areas of analysis and
mathematical physics. Moreover, a growing interest has been recently developed to the study of certain generalizations
of Toeplitz determinants. The most known among those are the determinants
of Toeplitz plus Hankel matrices --- see \cite{dik,BE17,GI}, the bordered 
Toeplitz determinants \cite{AYP87}, and the integrable Fredholm
determinants \cite{IIKS90,D-intop}. These determinants appear in the study of
Ising model in the zig-zag layered half-plane~\cite{CHM20}, in the spectral analysis of the
Hankel matrices, in the study of the next-to-diagonal correlation functions in the Ising model (\cite{AYP87}),
 and in the theory of exactly solvable quantum models. In this paper,  motivated 
 by the physical model in the context of quantum information, we are concerned with 
 yet another generalization  of Toeplitz determinants, which are 
 certain  finite rank deformations  of the standard Toeplitz matrices. In order to study
 such deformations, we need to analyse the asymptotic behaviour not only of the Toeplitz determinants {\it per se}
 but of the inverse Toeplitz matrices as well. The evaluation of the asymptotic behavior  of this new class of 
 determinants which is done in this paper  is, we believe, an important analytical result
 in its own right. 
 
To summarize, in this article we compute the mutual information between 
a two blocks of Fermi operators separated by one lattice site and the rest of the chain in the Hamiltonian~\eqref{Fchain} explicitly.  Our approach is based on the Riemann-Hilbert method, which has the additional  advantage of being mathematically rigorous.\color{black}

\section{The main result}
Let $C$ denote the unit circle on the complex plane and 
\begin{equation*}
	g\colon C \to \C, \quad
	g(z) = \left\{ \begin{matrix}
		1 & \Re z > 0, \\ 
		-1 & \Re z < 0
	\end{matrix} \right..
\end{equation*}
The Fourier coefficients of $g$ are
\begin{equation*}
	g_l := \frac{1}{2 \pi} \int_{-\frac{\pi}{2}}^{\frac{3\pi}{2}} e^{-il\theta} g(e^{i\theta}) d\theta 
	= \oint_{C} z^{-l} g(z) \frac{dz}{2 \pi i z} = \frac{2}{l \pi} \sin \frac{l \pi}{2}
	= \left\{ \begin{matrix}
		0 & l \text{ is even}, \\ 
		(-1)^{\frac{l-1}{2}} \frac{2}{l \pi} & l \text{ is odd}
	\end{matrix} \right..
\end{equation*}
In general, the $m\times m$ Toeplitz matrix and determinant with symbol $\phi \in L^\infty(C)$ will be denoted by $T_m[\phi]$ and $D_m[\phi]$, respectively.
As it is well-known, the spectral norm (or operator norm) of $T_m[\phi]$ satisfies $\|T_m[\phi]\| \leq \|\phi\|_\infty$. 
In particular, as $T_m[g]$ is a self-adjoint matrix, we obtain the relation $\sigma(T_m[g])\subseteq [-1,1]$ for its spectrum.

Let $k,m,n\in\N$.
We introduce the following matrix and determinant
\begin{equation}\label{eq:A1}
	A = \left(\begin{matrix}
		A_{11} & A_{12} \\
		A_{21} & A_{22}
	\end{matrix}\right)\in\C^{(m+n)\times(m+n)},
	\qquad
	D(\lambda) = \det(\lambda I - A) \;\; (\lambda\in\C)
\end{equation}
where 
\begin{equation}\label{eq:A2}
	A_{11} = -T_m[g]\in\C^{m\times m}, \quad A_{22} = -T_n[g]\in\C^{n\times n}, \quad A_{12} = A_{21}^T = \left(\mathcal{A}_{ij}\right)_{i=1,\dots, m; j=1,\dots, n}\in\C^{m\times n},
\end{equation}
and
\begin{equation}\label{eq:calA-or}
	\mathcal{A}_{ij} = -\left|\begin{matrix}
		g_{i-j-m-k} & g_{i-m-1} & g_{i-m-2} & \dots & g_{i-m-k} \\
		g_{1-j-k} & g_0 & g_{-1} & \dots & g_{1-k} \\
		g_{2-j-k} & g_1 & g_0 & \dots & g_{2-k} \\
		\vdots & \vdots & \vdots & \ddots & \vdots \\
		g_{-j} & g_{k-1} & g_{k-2} & \dots & g_0
	\end{matrix}\right|,
\end{equation}
which is the determinant of a $(k+1)\times(k+1)$ matrix, $k\in\N$.

Define the quantity
\begin{equation}\label{eq:entr-int}
	S(\rho_P) = \lim_{\varepsilon\searrow 0} \frac{1}{2\pi i} \oint_{\Gamma_\varepsilon} e(1+\varepsilon,\lambda)\frac{d}{d\lambda}\ln D(\lambda) \; d\lambda,
\end{equation}
where 
\begin{equation*}
	e(x,v) := -\frac{x+v}{2}\ln\frac{x+v}{2}-\frac{x-v}{2}\ln\frac{x-v}{2}.
\end{equation*}
The contour $\Gamma_\varepsilon$ goes around the $[-1,1]$ interval once in the positive direction avoiding the cuts $(-\infty,-1-\varepsilon]\cup[1+\varepsilon,\infty)$ of $e(1+\varepsilon,\cdot)$, see Figure \ref{fig:Gamma_epsilon}.
For instance $\Gamma_\varepsilon$ can be the circle $(1+\frac{1}{2}\varepsilon)C$.
\begin{figure}[h]
	\includegraphics[scale=0.5]{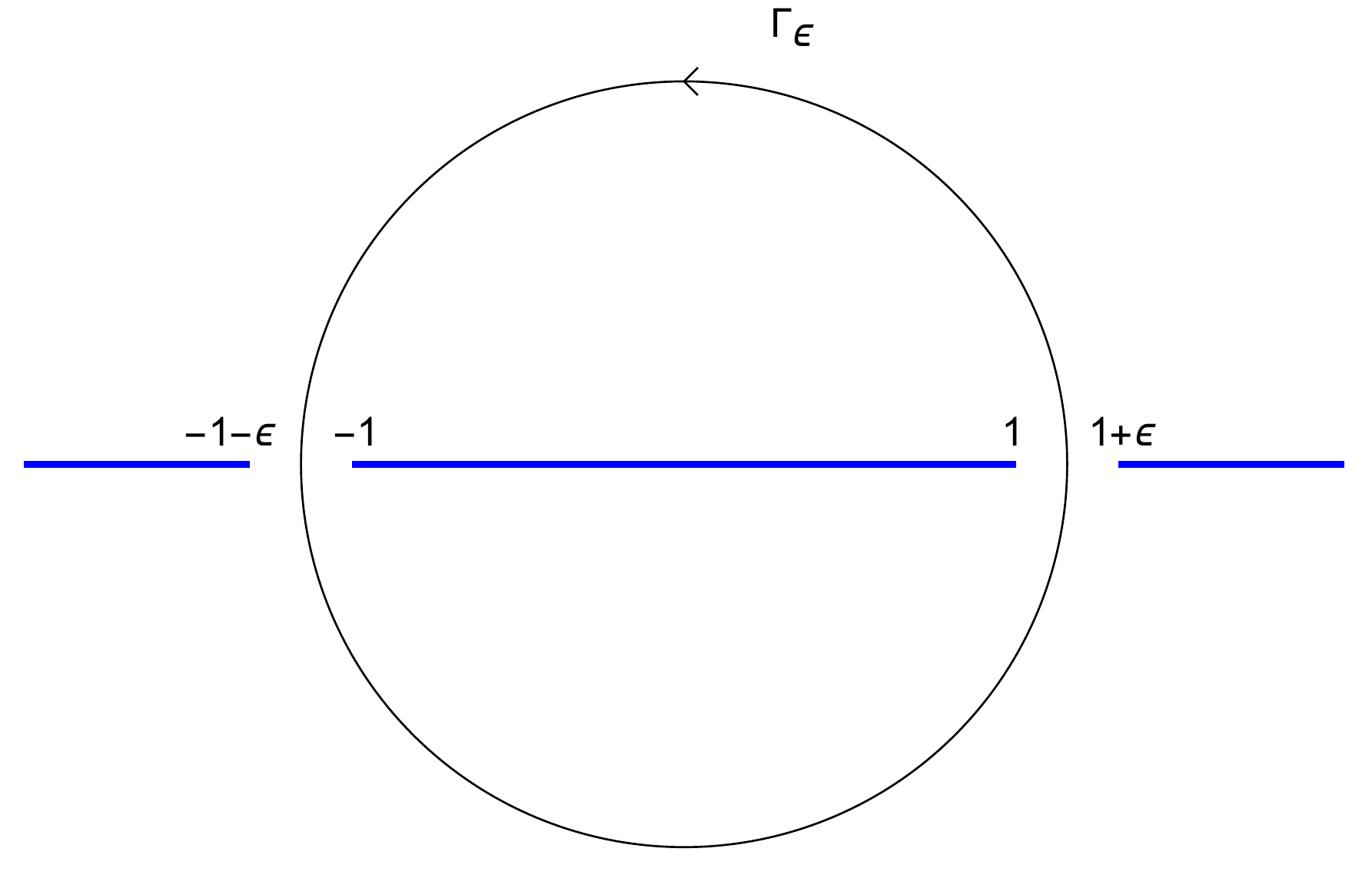}
	\caption{The cuts and the contour in (\ref{eq:entr-int}).}\label{fig:Gamma_epsilon}
\end{figure}
For a general $k, m, n$ we interpret the quantity in \eqref{eq:entr-int} as a measure of entanglement {(kind of an entropy)} between the subsystem 
\begin{equation}\label{eq:lat-m-k-n}
	P = \{1,2,\dots, m\}\cup\{m+k+1,m+k+2,\dots,m+k+n\}.
\end{equation}
and the rest of the chain of free fermions ~\eqref{Fchain} in the limit $N \to \infty$. 
Here is our motivation for this interpretation.

Let $\mathcal{H}$ be an Hilbert space spanned by the fermions in the chain ~\eqref{Fchain}.  Decompose $\mathcal{H}$ in the direct product $\mathcal{H}=\mathcal{H}_P \otimes \mathcal{H}_Q$, where $\mathcal{H}_P$ is the space generated by the fermions $b_j$ at the lattice sites $P$ indicated in ~ \eqref{eq:lat-m-k-n}. Write $P=P_1 \cup P_2$, where $P_1=\{1,\dotsc,m\}$ and  $P_2=\{m+k+1,\dotsc,m +k +n\}$ and  denote by $\ell_{P_1}$ and $\ell_{P_2}$ the sizes of $P_1$ and $P_2$,
respectively.  A standard calculation leads to the formula 
\begin{equation}
\rho_P = \frac{1}{2^{\ell_{P_1}+ \ell_{P_2}}} \sum_{a=0,1}\left\langle\left(\prod_{j\in P} b_j^a\right)\right\rangle\left(\prod_{j \in P}b_j^a\right)^\dagger
\end{equation}
for the reduced density matrix.  The angle brackets in this equation denote the expectation value with respect to the ground state.  Applying Wick's theorem gives
\begin{equation}
\label{Wicktheorem}
\rho_p = \prod_{j \in P} \left( \left \langle b_j^\dagger b_j \right \rangle b_j^\dagger b_j + \left \langle b_j b_j^\dagger\right 
\rangle b_j b_j^\dagger \right).
\end{equation}
The above subsystem consists of two blocks/intervals of $m$ and $n$ fermions separated by a gap of length $k$. Using ~\eqref{Wicktheorem}, 
it was shown in \cite{JK11} that in the special case when $k=m=n$, and in the thermodynamical limit $N \to \infty$,
the quantity   $S(\rho_P)$  is indeed the von Neumann entropy of \eqref{eq:lat-m-k-n}.  We refer the reader to \cite{JK11} for more details.

Our ultimate interest is to analyse $S(\rho_P)$ as $k,m,n\to\infty$, however, at this point the general problem seems to be far too complicated to attack directly
(see Remark \ref{rem1} in Section \ref{sec3} below).
Therefore we decided to start with the easier case when the gap between the two intervals is fixed to be $k=1$, that is, when \eqref{eq:lat-m-k-n} becomes
\begin{equation}\label{eq:lat-m-1-n}
	P = \{1,2,\dots, m\}\cup\{m+2,m+3,\dots,m+n+1\}.
\end{equation}
In this case the entries of $A_{12}$ in \eqref{eq:calA-or} become
\begin{equation*}
	\mathcal{A}_{ij} = -\left|\begin{matrix}
		g_{i-j-m-1} & g_{i-m-1} \\
		g_{-j} & g_0
	\end{matrix}\right|
\end{equation*}
and, taking into account that $g_0 =0$,
\begin{equation}\label{k1Aij}
\mathcal{A}_{ij} =  g_{i-m-1}\cdot g_{-j}.
\end{equation}
As we shall see, this simplest case already leads to a mathematically very challenging problem.

The asymptotic behaviour of the von Neumann entropy $S(\rho_P^{(n)})$ of the interval $\{1,2,\dots, n\}$ was calculated in \cite{JK04}.
In particular, it was shown there that
\begin{equation}\label{eq:as-int}
	S\left(\rho_P^{(n)}\right) = \lim_{\varepsilon\searrow 0} \frac{1}{2\pi i} \oint_{\Gamma_\varepsilon} e(1+\varepsilon,\lambda)\frac{d}{d\lambda}\ln D_n[\phi] \; d\lambda,
\end{equation}
where $\phi(z) = g(z)+\lambda$ $(z\in C)$.
Therefore the problem of calculating the limiting behaviour of the entropy of \eqref{eq:lat-m-1-n} reduces to the calculation of the \emph{mutual information} between the two intervals:
\begin{equation*}
	S(\rho_P^{(m)})+S(\rho_P^{(n)})-S(\rho_P) = \lim_{\varepsilon\searrow 0} \frac{1}{2\pi i} \oint_{\Gamma_\varepsilon} e(1+\varepsilon,\lambda)\frac{d}{d\lambda}\left( \ln D_m[\phi] + \ln D_n[\phi] - \ln D(\lambda) \right) \; d\lambda
\end{equation*}
To analyse the asymptotic behaviour of this quantity as $m,n\to\infty$ is still mathematically very complicated.
However, as we expect this quantity to converge to a finite number, it makes sense to consider the following limit instead, where the $\varepsilon$ and $m,n$ limits are interchanged:
\begin{equation}\label{eq:interchanged}
	\lim_{\varepsilon\searrow 0} \lim_{m,n\to\infty} \frac{1}{2\pi i} \oint_{\Gamma_\varepsilon} e(1+\varepsilon,\lambda)\frac{d}{d\lambda}\left( \ln D_m[\phi] + \ln D_n[\phi] - \ln D(\lambda) \right) \; d\lambda.
\end{equation}
We point out that a similar interchanged limit was considered in \cite{IJK,JK04} for the case of one interval. 
The value of the limit \eqref{eq:interchanged} is what we shall calculate and interpret as the mutual information between the two intervals.
It will turn out that indeed this is a finite number, which is stated in our main theorem.

\begin{theorem}\label{thm:main}
	Let $\widehat{D}(\lambda) = \frac{D(\lambda)}{D_m[\phi]\cdot D_n[\phi]}$ $(\lambda\in\C\setminus [-1,1])$.
	The limiting mutual information between the two intervals of the subsystem $P$ from \eqref{eq:lat-m-1-n} is
	\begin{equation}\label{eq:mut-inf}
	\lim_{\varepsilon\searrow 0} \lim_{m,n\to\infty} \frac{-1}{2\pi i} \oint_{\Gamma_\varepsilon} e(1+\varepsilon,\lambda)\frac{d}{d\lambda}\ln \widehat{D}(\lambda) \; d\lambda = \color{black} 2\ln 2 -1 \approx 0.386294.
	\end{equation}
\end{theorem}

\textcolor{black}{The main tool in the proof of the above theorem will be an asymptotic analysis of an inner product involving the inverse Toeplitz matrix $T_m[\phi]^{-1}$. We phrase the related statement in the next section as Lemma \ref{lem:main}.}

\section{Some preliminary calculations}\label{sec3}
We introduce the notations
\begin{align*}
	\vecg = (g_{-m}, g_{-m+1},\dots, g_{-1})^T \in \C^{m}, &\qquad \vecgg = (g_{-1}, g_{-2},\dots, g_{-n})^T \in \C^{n}, \\
	\vecG = T_m[\phi]^{-1}\vecg \in \C^{m}, &\qquad \vecGG = T_n[\phi]^{-1}\vecgg  \in \C^{n}.
\end{align*}
Notice that for all $\lambda\in\C\setminus [-1,1]$ we have (see also (\ref{k1Aij}))
\begin{equation*}
	\lambda I - A 
	= \left(\begin{matrix}
		T_m[\phi] & -\vecg \vecgg^T \\
		-\vecgg \vecg^T & T_n[\phi]
	\end{matrix}\right)
	= \left(\begin{matrix}
		T_m[\phi] & 0 \\
		0 & T_n[\phi]
	\end{matrix}\right)
	\cdot \left[
	I-
	\left(\begin{matrix}
		0 & \vecG \vecgg^T \\
		\vecGG \vecg^T & 0
	\end{matrix}\right)
	\right].
\end{equation*}
Therefore we obtain
\begin{align*}
	D(\lambda)
	= D_m[\phi] \cdot D_n[\phi] \cdot \det\left(I - \vecG\vecgg^T\vecGG\vecg^T\right) 
	= D_m[\phi] \cdot D_n[\phi] \cdot \left(1 - \langle\vecG,\vecg\rangle\langle\vecGG,\vecgg\rangle\right)
\end{align*}
where we used standard facts about rank-one matrices and the following identity for block-matrices:
\begin{equation*}
	I-\left(\begin{matrix}
		0 & B \\
		C & 0
	\end{matrix}\right) 
	= \left(\begin{matrix}
		I-BC & -B \\
		0 & I
	\end{matrix}\right)
	\cdot\left(\begin{matrix}
		I & 0 \\
		-C & I
	\end{matrix}\right).
\end{equation*}
In particular, we infer
\begin{equation*}
	\widehat{D}(\lambda) = 1 - \left\langle\vecG,\vecg\right\rangle\left\langle\vecGG,\vecgg\right\rangle.
\end{equation*}
Thus in order to compute the mutual information, we need to deal with the inner products $\langle\vec{\mathfrak{G}_j},\vec{\mathfrak{g}_j}\rangle$.
It turns out that it is sufficient to handle the case $j=1$. 

\begin{proposition}
	Let us use the notations $\vec{\mathfrak{g}}_1^{(m)}=\vec{\mathfrak{g}}_1$, $\vec{\mathfrak{G}}_1^{(m)}=\vec{\mathfrak{G}}_1$, $\vec{\mathfrak{g}}_2^{(n)}=\vec{\mathfrak{g}}_2$, $\vec{\mathfrak{G}}_2^{(n)}=\vec{\mathfrak{G}}_2$, which indicates the $m$- or $n$-dependence of the vectors.
	Then, we have
	\begin{equation*}
	\left\langle \vec{\mathfrak{G}}_2^{(n)}, \vec{\mathfrak{g}}_2^{(n)}\right\rangle = \left\langle \vec{\mathfrak{G}}_1^{(n)}, \vec{\mathfrak{g}}_1^{(n)} \right\rangle.
	\end{equation*}
\end{proposition}

\begin{proof}
	Consider the $n\times n$ matrix 
	$J = (\delta_{i+j-n+1})_{i,j=0}^{n-1}$ where $\delta$ denotes the Kronecker delta symbol.
	Since we have $J T_n[g] J = T_n[g]$, we obtain $J T_n[\phi]^{-1} J = T_n[\phi]^{-1}$, and thus
	\begin{equation*}
		\left\langle \vec{\mathfrak{G}}_2^{(n)}, \vec{\mathfrak{g}}_2^{(n)}\right\rangle 
		= \left\langle J T_n[\phi]^{-1} J\vec{\mathfrak{g}}_2^{(n)}, \vec{\mathfrak{g}}_2^{(n)}\right\rangle
		= \left\langle T_n[\phi]^{-1} \vec{\mathfrak{g}}_1^{(n)}, \vec{\mathfrak{g}}_1^{(n)}\right\rangle.
	\end{equation*}
\end{proof}

\begin{remark}\label{rem1} Notice that for a general gap of length  $k$, the matrix $\lambda I - A$ whose determinant
is needed to be evaluated, can be written as 
\begin{equation}\label{A tilde mat}
\lambda I - A = \begin{pmatrix} (\phi_{i-j})_{i,j=1,\dots , m} & \gamma\left(g_{i-j-m-1}\right)_{\substack{i= 1, \dots , m \\ j = 1,\dots n}} \\ \gamma\left(g_{j-i-m-1}\right)_{\substack{i= 1, \dots , n \\ j = 1,\dots m}} &  (\phi_{i-j})_{i,j=1,\dots , n} \end{pmatrix} 
- \sum_{d,l=1}^{k}\gamma_{dl}\begin{pmatrix} 0 & -\vec{\mathfrak{g}}_d \vec{\mathfrak{g}}_{k+l}^T  \\ -\vec{\mathfrak{g}}_{k+l} \vec{\mathfrak{g}}_{d}^T& 0 \end{pmatrix}
\end{equation} 
where
$$
\vec{\mathfrak{g}}_d= (g_{1-m-d}, g_{2-m-d},\dots, g_{-d})^T, \quad \vec{\mathfrak{g}}_{k+d}= (g_{d-1-k}, g_{d-2-k},\dots, g_{d-n-k})^T,
\quad d = 1, 2, ..., k,
$$
and the scalar  coefficients $\gamma$ and $\gamma_{dl}$ are certain $k\times k$,   independent of $n$, $m$,  determinants. 
This shows  what are the new technical challenges when one passes from $k=1$  to the values  $k >1$.  The  ``principal''
determinant is not  a block diagonal Toeplitz determinant anymore; indeed, the non-trivial off-diagonal Toeplitz blocks, generated by 
new symbols, appear. Moreover, the finite rank perturbation is of rank $2k$ and, therefore,   ceases to be ``finite rank''
as we consider the most general setting of the problem when all three sizes, $m$, $n$, and  $k$  become arbitrarily large. 
\end{remark}

\color{black} From now on, until the end of Section \ref{sec:before-final}, our goal is to prove the following lemma, which then we shall apply in Section \ref{sec:final} to prove Theorem \ref{thm:main}.

\begin{lemma}\label{lem:main}
	Define $\beta := \frac{1}{2 \pi i} \ln \frac{\lambda+1}{\lambda-1}$.
	As $m\to\infty$ we have
	\begin{equation}\label{eq:inner-prod-as}
	\langle\vecG,\vecg\rangle = \langle T_m[\phi]^{-1}\vecg,\vecg\rangle = i\tan\left(\tfrac{\pi}{2}\beta\right) + O(m^{-\frac{1}{4}}),
	\end{equation}
	where the error term is uniform in $\lambda$ on compact subsets of $|\lambda|>1$.
\end{lemma}

\color{black} In order to analyse $\langle\vecG,\vecg\rangle$, we shall express it in terms of a Riemann--Hilbert problem (RHP) that arises in the theory of integrable operators, see \cite[Section 5.6]{BDS} \textcolor{black}{or \cite{D-intop}.}
Define the kernel
\begin{equation*}
	K(z,s) = \frac{1-\phi(s)}{2\pi i} \frac{z^ms^{-m}-1}{z-s} = \frac{\vec{f}(z)^T\vec{h}(s)}{z-s} \qquad (z,s\in C),
\end{equation*}
where
\begin{equation*}
	\vec{f}(z) = \left(\begin{matrix} f_1(z) \\ f_2(z) \end{matrix}\right) = \left(\begin{matrix}z^m \\ 1 \end{matrix}\right), \qquad
	\vec{h}(s) = \left(\begin{matrix} h_1(s) \\ h_2(s) \end{matrix}\right) = \frac{1-\phi(s)}{2\pi i} \left(\begin{matrix}s^{-m} \\ -1 \end{matrix}\right),
\end{equation*}
which also satisfy $\left\langle\vec{f}(z),\vec{h}(z)\right\rangle = 0$ $(z\in C)$, where $\langle\vec{a},\vec{b}\rangle = \sum_{j} a_jb_j$.
This kernel defines a very special type of bounded singular integral operators on $L^2(C)$, namely a so-called (completely) \emph{integrable operator} in the following way:
\begin{equation*}
	K[u](z) = \int_C K(z,s)u(s) ds \qquad (u\in L^2(C), z\in C),
\end{equation*}
where the integral is meant in the principal value sense, and we put the function in between $[\;]$.

By well-known properties of this operator, for all $\lambda\in\C\setminus[-1,1]$ we have $0\neq D_m[\phi] = \det(1-K)$ and
\begin{equation}\label{eq:Tinv-K}
	T_m[\phi]^{-1} = \left(\left( (1-K)^{-1}[z^j], z^k \right)\right)_{j,k=0}^{m-1},
\end{equation}
where $(\cdot,\cdot)$ denotes the complex inner product on $L^2(C)$.
\textcolor{black}{In particular, the connection between the two determinants can be shown by repeating the argument of \cite[page 123]{BDS}. In order to obtain \eqref{eq:Tinv-K} we observe that (by (5.157)-(5.158) in \cite[page 123]{BDS}) $1-K$ has the block-matrix form $\left(\begin{matrix} 1 & 0 & 0 \\ * & T_m[\phi] & * \\ 0 & 0 & 1 \end{matrix}\right)$, hence the (2,2) block of $(1-K)^{-1}$ is $T_m[\phi]^{-1}$.}
Furthermore, by \cite[Theorem 5.21]{BDS} \textcolor{black}{or \cite{D-intop}}
\begin{equation}\label{eq:F}
	\vec{F}(z) = \left(\begin{matrix} F_1(z) \\ F_2(z) \end{matrix}\right) := (1-K)^{-1}[\vec{f}](z) = Y_{K-}(z) \vec{f}(z) 
	= \left(\begin{matrix}
		Y_{K-,11}(z)z^m+Y_{K-,12}(z) \\
		Y_{K-,21}(z)z^m+Y_{K-,22}(z)
	\end{matrix}\right),
\end{equation}
where $Y_K$ is the unique solution of the following RHP.

\bigskip

\noindent \textbf{$Y_K$--Riemann--Hilbert problem}
\begin{align}
	& Y_K\colon \C\setminus C \to \C^{2\times 2} \;\; \text{is analytic}, \\
	& Y_{K+}(z) = Y_{K-}(z)\cdot 
	\left(\begin{matrix}
		\phi(z) & -(\phi(z)-1)z^m \\
		(\phi(z)-1)z^{-m} & 2-\phi(z)
	\end{matrix}\right) \qquad (\text{a.e.} \; z\in C), \label{eq:YK-jump}\\
	& Y_K(z) = I + O(z^{-1}) \;\;\text{as} \;\; z\to\infty. \label{eq:YK-infty}
\end{align}
The unit circle is oriented in the usual positive direction, and the jump condition \eqref{eq:YK-jump} is meant in the $L^2$ sense, see \cite[Definition 5.16]{BDS}.

\bigskip

\noindent In the next section we will connect the $Y_K$--RHP with \textcolor{black}{another RHP}, but for the rest of this section our aim is to express the inner product $\langle\vecG,\vecg\rangle$ in terms of $Y_K$.

\begin{proposition}
	We have 
	\begin{equation}\label{eq:inn-M}
	\left\langle\vecG,\vecg\right\rangle = -\frac{2}{\pi}\sum_{\ell=1}^m \frac{\sin\frac{\ell\pi}{2}}{\ell}M_{\ell,11},
	\end{equation}
	where
	\begin{equation*}
	Y_K(z) = I + \sum_{\ell=1}^{\infty} M_\ell z^{-\ell} \qquad \text{as} \;\; z\to\infty.
	\end{equation*}
\end{proposition}

\begin{proof}
First, by \eqref{eq:Tinv-K} we have
\begin{align*}
	\langle\vecG,\vecg\rangle & = \vecg^T T_m[\phi]^{-1} \vecg = \sum_{j,k=0}^{m-1} g_{j-m} g_{k-m} \left( (1-K)^{-1}[z^j], z^k \right) \\
	& = \left( (1-K)^{-1}\left[\sum_{j=0}^{m-1} g_{j-m}z^j\right], \sum_{k=0}^{m-1} g_{k-m}z^k \right).
\end{align*}
Next, since $g(z) = (\phi(z)-1)+(1-\lambda)$, we calculate
\begin{align*}
	\sum_{j=0}^{m-1} g_{j-m}z^j & = \sum_{j=0}^{m-1} \oint_C g(s)s^{m-j}z^j \frac{ds}{2\pi i s} = \sum_{j=0}^{m-1} \oint_C (\phi(s)-1)s^{m-j}z^j \frac{ds}{2\pi i s} \\
	& = \oint_C \frac{\phi(s)-1}{2\pi i} s^m \sum_{j=0}^{m-1} \left(\frac{z}{s}\right)^j \frac{ds}{s} = - \oint_C K(z,s) s^m ds = -K[z^m](z).
\end{align*}
Therefore, by \eqref{eq:F} and \eqref{eq:YK-infty}
\begin{align*}
	\langle\vecG,\vecg\rangle & = - \left( (1-K)^{-1}K[z^m], \sum_{k=0}^{m-1} g_{k-m}z^k \right) = - \left( (1-K)^{-1}[z^m] - z^m, \sum_{k=0}^{m-1} g_{k-m}z^k \right) \\
	& = - \left( F_1, \sum_{k=0}^{m-1} g_{k-m}z^k \right) = -\sum_{k=0}^{m-1} g_{k-m} \oint_C F_1(z) z^{-k} \frac{dz}{2\pi i z} \\
	& = -\sum_{k=0}^{m-1} g_{k-m} \oint_C Y_{K-,11}(z)z^{m-k}+Y_{K-,12}(z)z^{-k} \frac{dz}{2\pi i z} \\
	& = -\sum_{k=0}^{m-1} g_{k-m} \oint_C Y_{K-,11}(z)z^{m-k} \frac{dz}{2\pi i z} = -\sum_{k=0}^{m-1} g_{k-m} M_{m-k,11},
\end{align*}
from which we conclude \eqref{eq:inn-M}.
\end{proof}

\section{Expressing the inner product in terms of the $R$--Riemann--Hilbert problem}
Note that for all $\lambda\in\C\setminus[-1,1]$ the function $\phi$ possesses Fisher--Hartwig singularities at $z_1 = i = e^{i \frac{\pi}{2}}$ and $z_2 = -i = e^{i \frac{3\pi}{2}}$; thus, we can apply the results in \cite{dik}.
To be more precise, using the notation of (1.2) in \cite{dik}, we can write $\phi$ in the following form:
\begin{equation*}
	\phi(z) = e^{V_0} g_1(z) g_2(z) z_1^{-\beta_1} z_2^{-\beta_2}
\end{equation*}
with $\alpha_1 = \alpha_2 = 0$, $\theta_1 = \frac{\pi}{2}$, $\theta_1 = \frac{3\pi}{2}$,
\begin{align}
	&\beta = \beta(\lambda) := \beta_1 = -\beta_2 = \frac{1}{2 \pi i} \cdot \ln \frac{\lambda+1}{\lambda-1} = \frac{1}{2 \pi i} \left[\ln(\lambda+1) - \ln(\lambda-1)\right], \label{eq:beta-lambda}\\
	&V(z) = V_0 = \frac{1}{2} \left[\ln(\lambda-1) + \ln(\lambda+1)\right], \\
	&g_1(z) g_2(z) = \left\{\begin{matrix}
		1 & \Re z > 0 \\
		e^{-2 i \pi \beta} & \Re z < 0
	\end{matrix}\right.
	= \left\{\begin{matrix}
		1 & \Re z > 0 \\
		\frac{\lambda-1}{\lambda+1} & \Re z < 0
	\end{matrix}\right. \\
	&z_1^{-\beta_1}z_2^{-\beta_2} = e^{\frac{1}{2} [\ln(\lambda+1) - \ln(\lambda-1)]}.
\end{align}
Note that throughout this paper, $\ln z$ denotes the principal branch of the logarithm, that is, $-\pi < \arg z < \pi$.
Since $\frac{\lambda+1}{\lambda-1}$ is a fractional linear map, we can easily examine the real- and imaginary parts of $\beta$.
We have
\begin {equation}\label{eq:Re-Im-beta}
	\Re \beta = \frac{1}{2\pi} \arg \frac{\lambda+1}{\lambda-1}, \qquad \Im \beta = \frac{-1}{2\pi} \ln \left|\frac{\lambda+1}{\lambda-1}\right| =  \frac{1}{2\pi} \ln \left|\frac{\lambda-1}{\lambda+1}\right|,
\end{equation}
therefore we see that $\Im\beta$ stays bounded on compact subsets of $\C\setminus [-1,1]$.
In addition, $|\Re\beta| < \frac{1}{2}$ for all $\lambda \in \C\setminus [-1,1]$.
However, notice that a simple calculation gives that 
$$|\Re\beta| < \frac{1}{4} \;\;\iff\;\; |\lambda| > 1,$$ 
which is the reason why we shall take $\Gamma_\varepsilon = (1+\frac{1}{2}\varepsilon)C$ in \eqref{eq:mut-inf} in our calculations.
Let us also note that $\beta$ does not vanish on $\C$.

Next, we shall connect the $Y_K$--RHP \textcolor{black}{with the $Y$--RHP, see e.g. \cite{dik} or \cite{bdj} for details.}

\bigskip

\noindent \textbf{$Y$--Riemann--Hilbert problem \textcolor{black}{for orthogonal polynomials on the circle}}
\begin{align}
	& Y\colon \C\setminus C \to \C^{2\times 2} \;\; \text{is analytic}, \\
	& Y_{+}(z) = Y_{-}(z)\cdot 
	\left(\begin{matrix}
		1 & \phi(z)z^{-m} \\
		0 & 1
	\end{matrix}\right) \qquad (z\in C\setminus\{i,-i\}), \label{eq:Y-jump}\\
	& Y(z) = \left(I + O(z^{-1})\right) 
	\left(\begin{matrix}
		z^m & 0 \\
		0 & z^{-m}
	\end{matrix}\right) \;\;\text{as} \;\; z\to\infty, \\
	&Y(z) = \left(\begin{matrix}
		O(1) & O(\ln|z\mp i|) \\
		O(1) & O(\ln|z\mp i|)
	\end{matrix}\right) \;\;\text{as} \;\; z\to\pm i.
\end{align}
The jump condition \eqref{eq:Y-jump} is meant in the sense that $Y$ is continuous up to $C$ from both sides, except at the points $\pm i$.

\bigskip

It is well-known that this RHP has a unique solution which can be given in terms of orthogonal polynomials.
\textcolor{black}{An easy calculation shows the following connection between the unique solutions $Y_K$ and $Y$:}
\begin{equation}\label{eq:YK-Y}
	Y_K(z) = \left\{
		\begin{matrix}
			\sigma_3 Y(z) \sigma_3 \left(\begin{matrix} z^m & -1 \\ 1 & 0 \end{matrix}\right)^{-1}, & |z|<1 \\ 
			\sigma_3 Y(z) \sigma_3 \left(\begin{matrix} z^m & 0 \\ 1 & z^{-m} \end{matrix}\right)^{-1}, & |z|>1
		\end{matrix}
	\right.
\end{equation}
where $\sigma_3 = \left(\begin{matrix} 1 & 0 \\ 0 & -1 \end{matrix}\right)$ is the Pauli matrix.
\textcolor{black}{We point out that a similar connection was observed in \cite{bdj}.}
Note that even though the jump conditions \eqref{eq:YK-jump} and \eqref{eq:Y-jump} are meant in different ways, one verifies easily that indeed the above $Y_K$ solves the $Y_K$--RHP in the $L^2$ sense.
The advantage of involving $Y$ in our analysis is that we can use the powerful results of \cite{dik}, in particular, we can express our inner product in terms of the $R$--RHP which can be estimated effectively.
Let us recall the $R$--RHP next, whose associated contour $\Gamma_R$ is shown in Figure \ref{fig:Gamma_R}.
Notice that the circles $\partial U_1$ and $\partial U_2$ around $\pm i$ are oriented in the negative direction.
\begin{figure}[h]
	\includegraphics[scale=0.3]{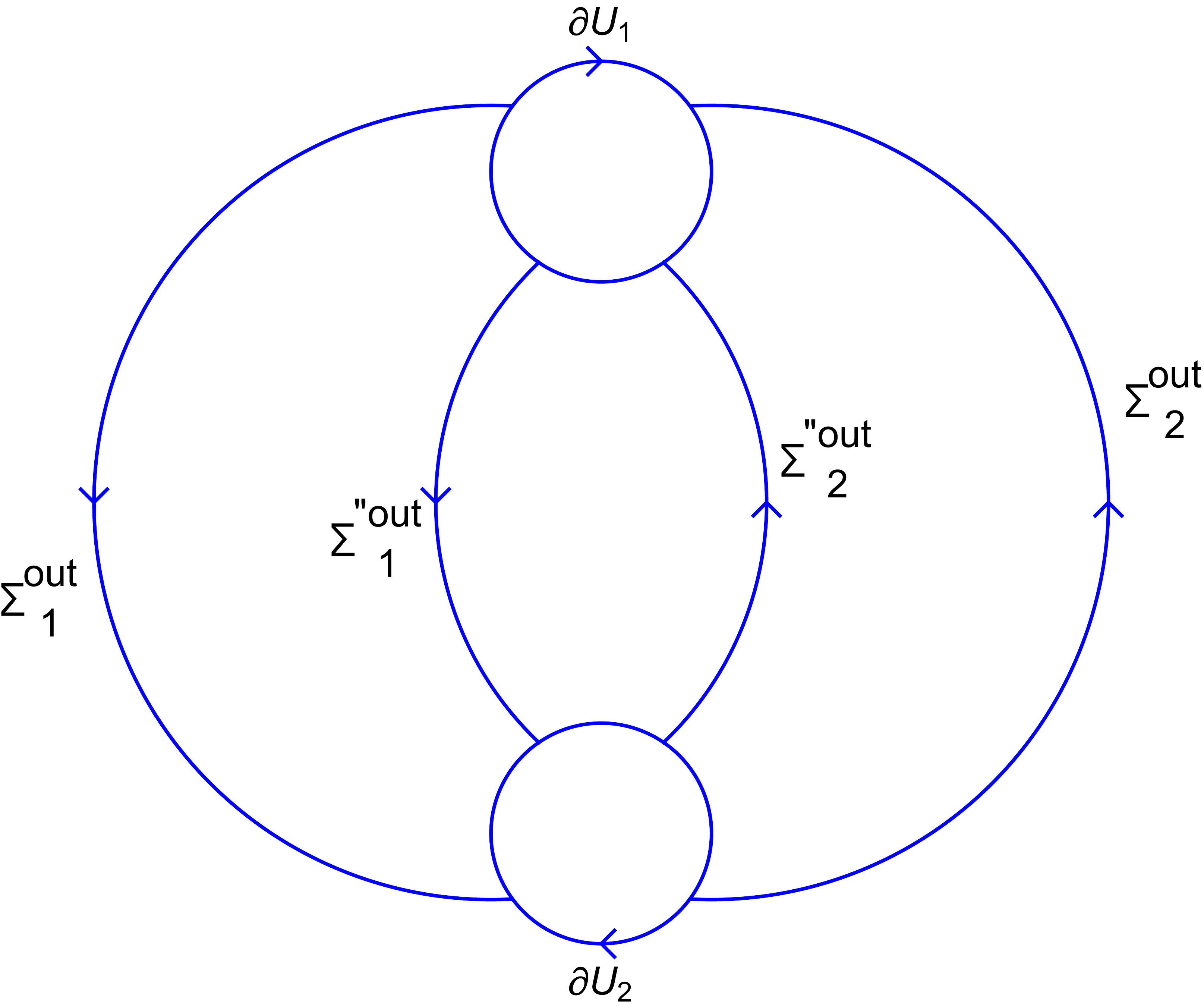}
	\caption{The contour $\Gamma_R$ for the $R$--RHP.}\label{fig:Gamma_R}
\end{figure}

\bigskip

\noindent \textbf{$R$--Riemann--Hilbert problem}
\begin{align}
	& R\colon \C\setminus\Gamma_R \to \C^{2\times 2} \;\; \text{is analytic}, \\
	& R_{+}(z) = R_{-}(z)\cdot N(z)
	\left(\begin{matrix}
		1 & 0 \\
		\phi(z)^{-1}z^{-m} & 1
	\end{matrix}\right) N(z)^{-1} \qquad (z\in\Sigma_j^{\text{out}}), \label{eq:R-jump-1} \\
	& R_{+}(z) = R_{-}(z)\cdot N(z)
	\left(\begin{matrix}
		1 & 0 \\
		\phi(z)^{-1}z^{m} & 1
	\end{matrix}\right) N(z)^{-1} \qquad (z\in\Sigma_j^{''\text{out}}), \label{eq:R-jump-2} \\
	& R_{+}(z) = R_{-}(z)\cdot P_j(z) N(z)^{-1} \qquad (z\in\partial U_j\setminus\{\text{intersection points}\}), \label{eq:R-jump-3} \\
	& R(z) = I + O(z^{-1}) \;\;\text{as} \;\; z\to\infty.
\end{align}

\bigskip

The jump conditions \eqref{eq:R-jump-1}--\eqref{eq:R-jump-3} are meant in the sense that $R$ is continuous up to $\Gamma_R$ from each side.
The functions $N$ and $P_j$ denote the global and local parametrices, respectively, see \cite[Subsections 4.1--4.2]{dik}.
Namely, 
\begin{equation}\label{eq:N-def}
	N(z) = \left\{\begin{matrix}
		\calD(z)^{\sigma_3} & |z|>1 \\
		\calD(z)^{\sigma_3} \left(\begin{matrix} 0 & 1 \\ -1 & 0 \end{matrix}\right) & |z|<1
	\end{matrix}\right.
\end{equation}
where $\calD(z) = \exp\left( \frac{1}{2\pi i} \int_C \frac{\ln \phi(s)}{s-z} ds\right)$ stands for the Szeg\H{o} function.
The local parametrices will be discussed in detail in Section \ref{sec:local-par}.

From \eqref{eq:YK-Y} we calculate
\begin{equation*}
	Y_{K,11}(z) = Y_{11}(z)z^{-m} + Y_{12}(z) \qquad (|z|\geq 2).
\end{equation*}
If we trace back the transformations $Y\to T\to S\to R$ performed in \cite{dik}, we obtain
\begin{equation*}
	Y(z) = R(z) N(z) z^{m\sigma_3} = R(z) \calD(z)^{\sigma_3} z^{m\sigma_3} \qquad (|z|\geq 2).
\end{equation*}
In particular,
\begin{equation*}
	Y_{11}(z)z^{-m} = R_{11}(z)\calD(z), \quad Y_{12}(z) = R_{12}(z)\calD(z)^{-1} z^{-m} \qquad (|z|\geq 2).
\end{equation*}
Notice that $Y_{12}(z) = O(z^{-m-1})$ as $z\to\infty$, hence by \eqref{eq:inn-M} it does not contribute to our inner product.
Therefore we have
\begin{equation*}
	M_{\ell,11} = d_\ell + \oint_{|z|= 2} \left(R_{11}(z)-1\right) \calD(z) z^\ell \frac{dz}{2\pi i z} \qquad (\ell = 1,\dots,m),
\end{equation*}
where $\calD(z) = 1 + \sum_{j=1}^\infty d_j z^{-j}$ $(|z|>1)$.
Thus, from \eqref{eq:inn-M} we obtain
\begin{equation}\label{eq:contributions0}
	\left\langle\vecG,\vecg\right\rangle = -\frac{2}{\pi}\sum_{\ell=1}^m \frac{\sin\frac{\ell\pi}{2}}{\ell}d_{\ell} -\frac{2}{\pi}\oint_{|z|= 2} \left(R_{11}(z)-1\right) \calD(z) f_m(z) \frac{dz}{2\pi i z},
\end{equation}
where
\begin{equation}\label{eq:fm-def}
	f_m(z) = \sum_{\ell=1}^m \frac{\sin\frac{\ell\pi}{2}}{\ell} z^\ell.
\end{equation}

Now, set $M=\lfloor\frac{m-1}{2}\rfloor$ and notice that for all $z\in\C, \Re z \neq 0$, we have
\begin{align}
	f_m(z) &= \sum_{k=0}^M \frac{(-1)^k}{2k+1} z^{2k+1} = \frac{1}{i}\sum_{k=0}^M \frac{(iz)^{2k+1}}{2k+1} = \frac{1}{i}\sum_{k=0}^M \int_{0}^{iz} s^{2k} ds = \frac{1}{i} \int_{0}^{iz} \frac{1-s^{2M+2}}{1-s^2} ds \nonumber \\
	&= \int_{0}^{z} \frac{1+(-1)^{M}y^{2M+2}}{1+y^2} dy = \arctan z + (-1)^{M}\int_{0}^{z} \frac{y^{2M+2}}{1+y^2} dy \nonumber \\
	&= \frac{1}{2i}\ln\frac{z-i}{z+i} + \frac{\pi}{2}\sgn\Re z + (-1)^{M}\int_{0}^{z} \frac{y^{2M+2}}{1+y^2} dy \nonumber \\
	&= \frac{1}{2i}\ln\frac{z-i}{z+i} + \widetilde{f_m}(z) = \frac{1}{2i}\ln\frac{z-i}{z+i} + \frac{\pi}{2}\sgn\Re z + \widetilde{f_{m,1}}(z), \label{eq:fm1-def}
\end{align}
where the integration is meant along a line segment and $\widetilde{f_m}$, $\widetilde{f_{m,1}}$ are implicitly defined in the above equation-chain.
Note that $\frac{1}{2i}\ln\frac{z-i}{z+i}$ and $\widetilde{f_m}$ are analytic in $\C\setminus[-i,i]$, and that the integral expression $\widetilde{f_{m,1}}$ is analytic in $\C\setminus ([i,i\infty)\cup[-i,-i\infty))$.
Since $\left(R_{11}(z)-1\right) \calD(z) = O(1/z)$ and $\frac{1}{2i}\ln\frac{z-i}{z+i} = O(1/z)$ as $z\to\infty$, we easily obtain that
\begin{equation}\label{eq:contributions}
	\left\langle\vecG,\vecg\right\rangle = -\frac{2}{\pi}\sum_{\ell=1}^m \frac{\sin\frac{\ell\pi}{2}}{\ell}d_{\ell} -\frac{2}{\pi}\oint_{|z|= 2} \left(R_{11}(z)-1\right) \calD(z) \widetilde{f_m}(z) \frac{dz}{2\pi i z},
\end{equation}

To summarise, we have two kinds of contributions to the inner product, one which comes from the Szeg\H{o} function and another coming from $R-I$.
Next, we compute the contribution coming from $\calD(z)$.

\section{The contribution from the Szeg\H{o} function}
Here we calculate the asymptotic behaviour of $-\frac{2}{\pi}\sum_{\ell=1}^m \frac{\sin\frac{\ell\pi}{2}}{\ell}d_{\ell} = -\sum_{\ell=1}^m g_{\ell}d_{\ell}$ as $m\to\infty$. 
For that, we need a formula for the Szeg\H{o} function.
By (4.8) (or (4.10)) in \cite{dik}, a short calculation gives
\begin{equation}\label{eq:D-def}
	\calD(z) = \exp\left( \beta\cdot\ln \frac{z-i}{z+i} \right)= \left(\frac{z-i}{z+i} \right)^\beta \qquad (|z|>1),
\end{equation}
where the right-hand side is analytic outside $[-i,i]$, and
\begin{equation}\label{eq:D-def-ins}
	\calD(z) = \left(\frac{z-i}{z+i} \right)^\beta \phi(z) \qquad (|z|<1).
\end{equation}
A simple calculation gives 
\begin{equation*}
	D\left(e^{i\theta}\right) = \left( \frac{e^{i\theta}-e^{i\frac{\pi}{2}}}{e^{i\theta}+e^{i\frac{\pi}{2}}} \right)^\beta = \left(i \frac{\sin\left(\frac{\theta}{2}-\frac{\pi}{4}\right)}{\cos\left(\frac{\theta}{2}-\frac{\pi}{4}\right)}\right)^\beta = \left( i \tan\left(\frac{\theta}{2}-\frac{\pi}{4}\right)\right)^\beta \qquad (\theta\in\R).
\end{equation*}
Therefore, since $\calD(1/z) = 1 + \sum_{j=1}^\infty d_j z^j$ belongs to the Hardy class $H^2$, we obtain the following expression for the limit of the sum:
\begin{align*}
	- \sum_{l=1}^\infty g_l d_{l} & = -\frac{1}{2\pi} \int_{-\frac{\pi}{2}}^{\frac{3\pi}{2}} \calD(e^{-i\theta}) \overline{g(e^{i\theta})} d\theta = -\frac{1}{2\pi} \int_{-\frac{\pi}{2}}^{\frac{3\pi}{2}} \calD(e^{i\theta}) g(e^{i\theta}) d\theta \\
	&= -\frac{1}{2\pi} \int_{-\frac{\pi}{2}}^{\frac{3\pi}{2}} \calD(e^{i\theta}) d\theta + \frac{1}{\pi} \int_{\frac{\pi}{2}}^{\frac{3\pi}{2}} \calD(e^{i\theta}) d\theta = -1 + \frac{1}{\pi} \int_{\frac{\pi}{2}}^{\frac{3\pi}{2}} \calD(e^{i\theta}) d\theta \\
	&= -1 + \frac{1}{\pi} \int_{\frac{\pi}{2}}^{\frac{3\pi}{2}} \left(i \tan\left(\frac{\theta}{2}-\frac{\pi}{4}\right)\right)^\beta d\theta
	= -1 + \frac{2 \cdot i^\beta}{\pi} \int_{0}^{\frac{\pi}{2}} \left(\tan\vartheta\right)^\beta d\vartheta \\
	& = -1 + \frac{2 \cdot i^\beta}{\pi} \int_{0}^{\infty} \frac{u^\beta}{1+u^2} du = \color{black} i\tan\left(\tfrac{\pi}{2}\beta\right),
\end{align*}
where we substituted $u = \tan\vartheta$ and used standard residue calculus.

We estimate the speed of convergence below.

\begin{proposition}
	We have
	\begin{equation}\label{eq:sigma-as}
	-\frac{2}{\pi}\sum_{\ell=1}^m \frac{\sin\frac{\ell\pi}{2}}{\ell}d_{\ell} = \color{black} i\tan\left(\tfrac{\pi}{2}\beta\right) + O(m^{-\frac{1}{2}}) 
	\end{equation}
	as $m\to\infty$, where the error is uniform in $\lambda$ on compact subsets of $\C\setminus[-1,1]$.
\end{proposition}

\begin{proof}
Note that
\begin{align*}
	|\calD(e^{i \theta})| 
	&= \left| \left( e^{i \pi / 2} \tan\left(\frac{\theta}{2}-\frac{\pi}{4}\right)\right)^\beta \right|
	= e^{-\frac{\pi}{2} \Im\beta} \left( \left|\frac{\theta-\frac{3\pi}{2}}{\theta-\frac{\pi}{2}} \right| \tan\left(\frac{\theta}{2}-\frac{\pi}{4}\right)\right)^{\Re\beta} \left|\theta-\frac{\pi}{2}\right|^{\Re\beta} \left|\theta-\frac{3\pi}{2}\right|^{-\Re\beta} \\
	&\leq e^{-\frac{\pi}{2} \Im\beta} 3^{1/2} \left|\theta-\frac{\pi}{2}\right|^{\Re\beta} \left|\theta-\frac{3\pi}{2}\right|^{-\Re\beta} \hspace{5cm} \left(\frac{\pi}{2} < \theta < \frac{3\pi}{2}\right).
\end{align*}
and similarly
\begin{align*}
	|\calD(e^{i \theta})| 
	&= \left| \left( e^{-i \pi / 2} \left|\tan\left(\frac{\theta}{2}-\frac{\pi}{4}\right)\right|\right)^\beta \right| 
	\leq e^{\frac{\pi}{2} \Im\beta} 3^{1/2} \left|\theta-\frac{\pi}{2}\right|^{\Re\beta} \left|\theta+\frac{\pi}{2}\right|^{-\Re\beta} & \qquad \left(-\frac{\pi}{2} < \theta < \frac{\pi}{2}\right).
\end{align*}
Hence the squared $L^2$-norm of the Szeg\H{o} function can be estimated as follows:
\begin{align*}
	\frac{1}{2\pi}\int_{-\frac{\pi}{2}}^{\frac{3\pi}{2}} |\calD(e^{i \theta})|^2 d\theta 
	&\leq \frac{e^{\pi |\Im\beta|} 3}{2\pi} \left( \int_{-\frac{\pi}{2}}^{\frac{\pi}{2}}  \left|\theta-\frac{\pi}{2}\right|^{2\Re\beta} \left|\theta+\frac{\pi}{2}\right|^{-2\Re\beta} d\theta
		+ \int_{\frac{\pi}{2}}^{\frac{3\pi}{2}} \left|\theta-\frac{\pi}{2}\right|^{2\Re\beta} \left|\theta-\frac{3\pi}{2}\right|^{-2\Re\beta} d\theta \right) \\
	&= \frac{e^{\pi |\Im\beta|} 3}{\pi} \int_{-\frac{\pi}{2}}^{\frac{\pi}{2}} \left|\theta-\frac{\pi}{2}\right|^{2|\Re\beta|}  \left|\theta+\frac{\pi}{2}\right|^{-2|\Re\beta|} d\theta
	\leq {e^{\pi |\Im\beta|} 3} \int_{0}^{\pi} t^{-2|\Re\beta|} d\theta \\
	&\leq 3 \frac{e^{\pi |\Im\beta|}\pi^{1-2|\Re\beta|}}{1-2|\Re\beta|}.
\end{align*}
Therefore by the Cauchy--Schwartz inequality we get
\begin{align*}
	& \left| \sum_{l=m+1}^\infty g_l d_{l} \right| 
	\leq \frac{2}{\pi} \sum_{l=m+1}^\infty \frac{1}{l} \left| d_{l} \right| 
	\leq \frac{2}{\pi} \sqrt{ \sum_{l=m+1}^\infty \frac{1}{l^2} } \sqrt{ \frac{1}{2\pi}\int_{-\pi}^{\pi} |\calD(e^{i \theta})|^2 d\theta } \nonumber \\
	&\leq \frac{2}{\pi} \sqrt{ \int_{m}^\infty \frac{1}{x^2} dx } \sqrt{ \frac{1}{2\pi}\int_{-\pi}^{\pi} |\calD(e^{i \theta})|^2 d\theta } \leq \frac{2}{\pi} m^{-1/2} e^{\frac{\pi}{2} |\Im\beta|} \sqrt{3}  \sqrt{\frac{1}{1-2|\Re\beta|}} \pi^{\frac{1}{2}-|\Re\beta|} = O(m^{-1/2})
\end{align*}
as $m\to\infty$, uniformly in $\lambda$ on compact subsets of $\C\setminus[-1,1]$.
\end{proof}

Before we proceed with computing the contribution coming from $R-I$, we need some auxiliary calculations about the integral $\widetilde{f_{m,1}}$ defined in \eqref{eq:fm1-def}, and the local parametrices appearing in the analysis of $R$ in \cite{dik}.

\section{Estimation of $\widetilde{f_{m,1}}$}

We start with the following proposition.

\begin{proposition}
	We have
		\begin{equation*}
			e^{-mu} = (1-u)^m + O\left(\frac{1}{m}\right)
		\end{equation*}
	as $m\to\infty$, \emph{uniformly} in $u\in[0,1]$.
\end{proposition}

\begin{proof}
	As $\ln(1-u) < -u$ $(0<u<1)$, we have $e^{-mu} > (1-u)^m$ $(0<u\leq 1)$.
	Note that 
	\begin{equation*}
		\frac{d}{du} (e^{-mu} - (1-u)^m) = m ((1-u)^{m-1} - e^{-mu}) = 0 \;\iff\; -mu = (m-1)\ln(1-u) \qquad (0 \leq u \leq 1).
	\end{equation*}
	Since $\ln(1-u)$ is concave, we have at most two stationary points, and clearly one of them is $u=0$.
	Simple calculations show that for $u=\frac{1}{m}$ the derivative is positive, and that for $u=\frac{3}{m}$ it is negative, therefore there is a second stationary point $\frac{1}{m} < \widetilde{u_m} < \frac{3}{m}$ ($m\in\N, m>5$).
	It is then obvious that for all $u\in[0,1]$ we have
	\begin{align*}
		0\leq e^{-mu} - (1-u)^m \leq e^{-m\widetilde{u_m}} - (1-\widetilde{u_m})^m = e^{-m\widetilde{u_m}}(1 - e^{m (\ln (1-\widetilde{u_m}) + \widetilde{u_m})}) \\
		= O(1)(1 - e^{O(m \widetilde{u_m}^2)}) = O(1)(1 - e^{O(1/m)}) = O\left(\frac{1}{m}\right)  
	\end{align*}
	as $m\to\infty$.
\end{proof}

We proceed with the estimation of $\widetilde{f_{m,1}}(z) = (-1)^{M}\int_{0}^{z} \frac{y^{2M+2}}{1+y^2} dy$ when $z$ is close to the cut $[-i,i]$.

\begin{lemma}\label{lem:fm-est1}
	As $m\to\infty$, we have the following estimates which are uniform in $z$ and $t$:
	\begin{itemize}
		\item[(i)] 
			\begin{equation*}
				\widetilde{f_{m,1}}(z) = O\left( 2^{-m} \right) \qquad (|z|\leq \frac{1}{2}),
			\end{equation*}
		\item[(ii)] 
			\begin{equation*}
				\widetilde{f_{m,1}}(it) = O\left(e^{-\frac{1}{2}\sqrt{m}}\frac{1}{\sqrt{m}}\right) \qquad (-1+\frac{1}{\sqrt{m}} \leq t \leq 1-\frac{1}{\sqrt{m}})
			\end{equation*}
		\item[(iii)] 
			\begin{equation*}
				\widetilde{f_{m,1}}(it) = \frac{1}{2i} \int_{m(1-t)}^{\infty} \frac{e^{-\zeta}}{\zeta} d\zeta + O\left(\frac{1}{m}\ln m\right) = O(1) \qquad (1-\frac{1}{\sqrt{m}} \leq t \leq 1-\frac{1}{m})
			\end{equation*}
			and $\widetilde{f_{m,1}}(-it) = O(1)$.
	\end{itemize}
\end{lemma}

\begin{proof}
	(i) is obvious. 
	Note that $\widetilde{f_{m,1}}$ is an odd function, therefore it is enough to prove (ii)--(iii) for $t\geq 0$.
	For $0\leq t \leq 1-\frac{1}{m}$ we have
	\begin{equation*}
		\widetilde{f_{m,1}}(it) = \frac{(-1)^M}{2i} \int_0^{it} \frac{y^{2M+2}}{y-i} dy - \frac{(-1)^M}{2i} \int_0^{it} \frac{y^{2M+2}}{y+i} dy
		= \frac{(-1)^M}{2i} I_1(t) - \frac{(-1)^M}{2i} I_2(t),
	\end{equation*}
	where $I_1(t)$ and $I_2(t)$ are the first and second integrals, respectively.
	By substituting $y = ix$ and keeping in mind that $m \in \{2M+1, 2M+2\}$, we get for all $0\leq t \leq 1-\frac{1}{m}$ that
	\begin{equation*}
		|I_2(t)| = \int_0^{t} \frac{x^{2M+2}}{x+1} dx \leq \int_0^{t} x^{2M+2} dx = \frac{t^{2M+3}}{2M+3} \leq \frac{(1-\frac{1}{m})^{2M+3}}{2M+3} \leq \frac{(1-\frac{1}{m})^m}{m} = O\left(\frac{1}{m}\right)
	\end{equation*}
	as $m\to\infty$. We also obtain that if $0\leq t\leq 1-\frac{1}{\sqrt{m}}$, then
	\begin{equation*}
		|I_2(t)| \leq |I_1(t)| = \int_0^{t} \frac{x^{2M+2}}{1-x} dx \leq \sqrt{m} \int_0^{1-\frac{1}{\sqrt{m}}} x^{2M+2} dx \leq \frac{1}{\sqrt{m}} \left(1-\frac{1}{\sqrt{m}}\right)^{m} = O\left( \frac{1}{\sqrt{m}} e^{-\frac{1}{2}\sqrt{m}} \right)
	\end{equation*}
	as $m\to\infty$, which proves (ii).
	
	Now, let $1-\frac{1}{\sqrt{m}} \leq t \leq 1-\frac{1}{m}$, then 
	\begin{align*}
		I_1(t) &= (-1)^M \int_0^{t} \frac{x^{2M+2}}{1-x} dx = (-1)^M \int_{1-\frac{1}{\sqrt{m}}}^{t} \frac{x^{2M+2}}{1-x} dx + O\left( \frac{1}{\sqrt{m}} e^{-\frac{1}{2}\sqrt{m}} \right) \\
		&= (-1)^M \int_{1-t}^{\frac{1}{\sqrt{m}}} \frac{(1-u)^{2M+2}}{u} du + O\left( \frac{1}{\sqrt{m}} e^{-\frac{1}{2}\sqrt{m}} \right) \\
		&= (-1)^M \int_{1-t}^{\frac{1}{\sqrt{m}}} \frac{e^{-(2M+2)u} + O(\frac{1}{m})}{u} du + O\left( \frac{1}{\sqrt{m}} e^{-\frac{1}{2}\sqrt{m}} \right) \\
		&= (-1)^M \int_{1-t}^{\frac{1}{\sqrt{m}}} \frac{e^{-(2M+2)u}}{u} du + O(\frac{1}{m}) \int_{1-t}^{\frac{1}{\sqrt{m}}} \frac{1}{u} du+ O\left( \frac{1}{\sqrt{m}} e^{-\frac{1}{2}\sqrt{m}} \right) \\
		&= (-1)^M \int_{1-t}^{\frac{1}{\sqrt{m}}} \frac{e^{-(2M+2)u}}{u} du + O(\frac{1}{m} \ln m) = (-1)^M \int_{(2M+2)(1-t)}^{\frac{2M+2}{\sqrt{m}}} \frac{e^{-\zeta}}{\zeta} d\zeta + O(\frac{1}{m} \ln m)
	\end{align*}
	as $m\to\infty$, uniformly in $t$.
	What remains to prove is that the latter integral is $\int_{m(1-t)}^{\infty} \frac{e^{-\zeta}}{\zeta} d\zeta + O(\frac{1}{m}\ln m)$ as $m\to\infty$, uniformly in $t$, which follows from the following calculations:
	\begin{align*}
		\int_{\frac{2M+2}{\sqrt{m}}}^{\infty} \frac{e^{-\zeta}}{\zeta} d\zeta 
		\leq \int_{\frac{2M+2}{\sqrt{m}}}^{\infty} e^{-\zeta} d\zeta
		= e^{-\frac{2M+2}{\sqrt{m}}}
		\leq e^{-\sqrt{m}}
	\end{align*}
	and
	\begin{align*}
		\int_{m(1-t)}^{(2M+2)(1-t)} \frac{e^{-\zeta}}{\zeta} d\zeta 
		= (2M+2-m)(1-t) \frac{e^{-m(1-t)}}{m(1-t)} \leq \frac{e^{-1}}{m}.
	\end{align*}
\end{proof}

Now, we estimate near $i$.

\begin{lemma}\label{lem:fm-est2}
	Let $0<c<1<C$, then the following holds as $m\to\infty$, uniformly in $\frac{c}{m} \leq |z-i| \leq \frac{C}{m}$, $z\notin [i,i\infty)$:
	\begin{equation}\label{eq:circ-est}
		\widetilde{f_{m,1}}(z) = \frac{1}{2i} \int_1^\infty \frac{e^{-\zeta}}{\zeta} d\zeta - \frac{1}{2i} \int_1^{im(z-i)} \frac{e^{-\zeta}}{\zeta} d\zeta + O\left(\frac{1}{m} \ln m\right) = O(1),
	\end{equation}
	where the path for the second integral lies in $\zeta \in \C\setminus (-\infty,0]$, $c \leq |\zeta| \leq C$, as shown in Figure \ref{fig:annulus}.
\end{lemma}
\begin{figure}[h]
	\includegraphics[scale=0.25]{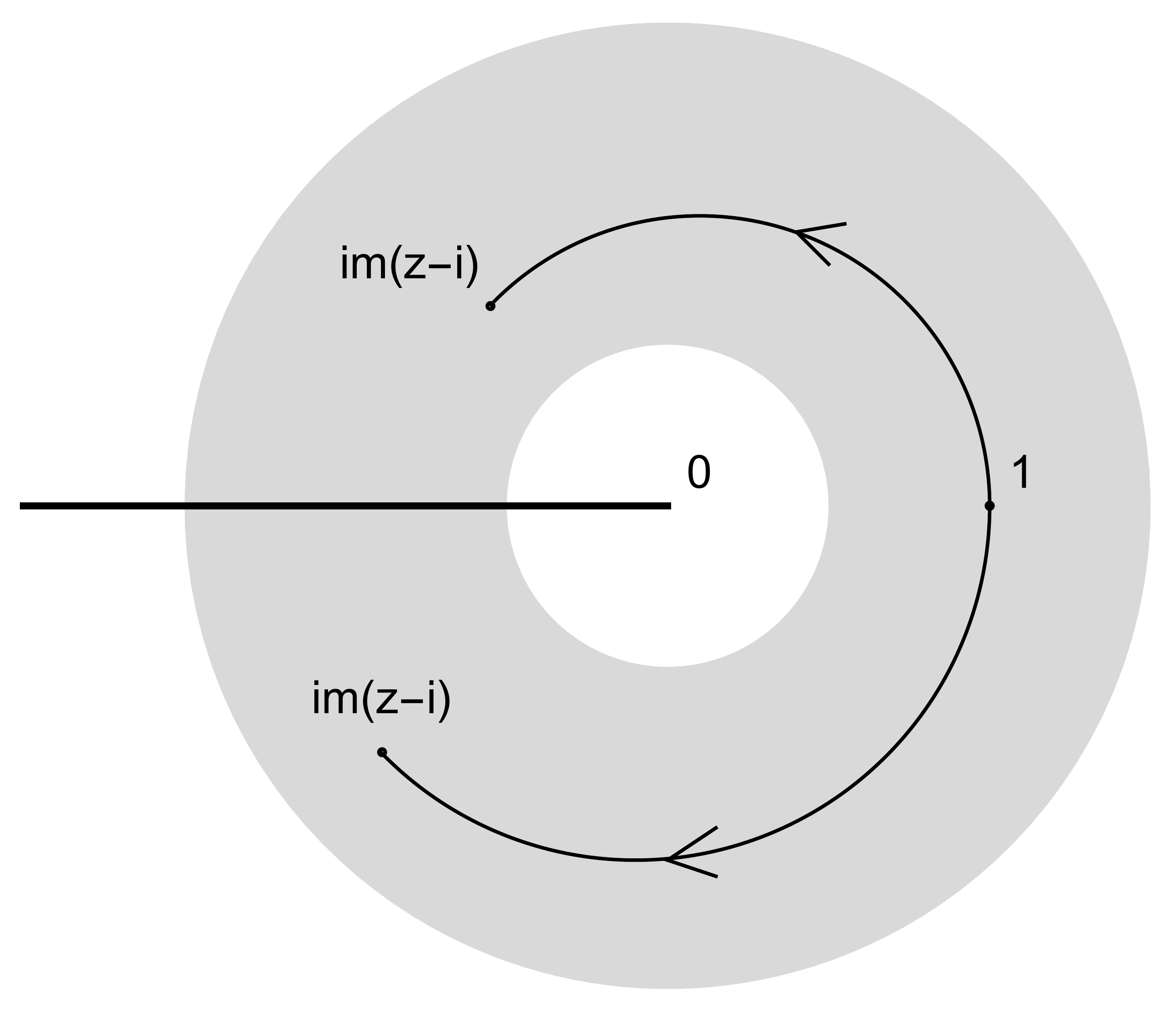}
	\caption{The contour of integration in (\ref{eq:circ-est}).}\label{fig:annulus}
\end{figure}
\begin{proof}
	Notice that for all $c<|u|<C$ we have 
	\begin{align*}
		(1-\frac{u}{m})^{2M+2} - e^{-u} &= e^{(2M+2) \ln(1-u/m)} - e^{-u} = e^{(2M+2) (-u/m + O(m^{-2}))} - e^{-u} \\
		&= e^{-u} \left( e^{-(2M+2-m)u/m + (2M+2) O(m^{-2})} - 1 \right)
		= e^{-u} \left( e^{O(1/m)} - 1 \right) = O\left(\frac{1}{m}\right).
	\end{align*}
	Therefore, using the substitution $y = i(1-\frac{u}{m})$, $u = im(y-i)$, we obtain
	\begin{align*}
		\widetilde{f_{m,1}}(z) 
		&= \widetilde{f_{m,1}}(i(1-\frac{1}{m})) + (-1)^M \int_{i(1-\frac{1}{m})}^z \frac{y^{2M+2}}{1+y^2} dy \\
		&= \frac{1}{2i} \int_1^{\infty} \frac{e^{-\zeta}}{\zeta} d\zeta + O\left(\frac{1}{m}\ln m\right) + \frac{(-1)^M}{2i} \int_{i(1-\frac{1}{m})}^z \frac{y^{2M+2}}{y-i} + O(1) dy \\
		&= \frac{1}{2i} \int_1^{\infty} \frac{e^{-\zeta}}{\zeta} d\zeta 
		- \frac{1}{2i} \int_{1}^{im(z-i)} \frac{(1-\frac{u}{m})^{2M+2}}{u} du + O\left(\frac{1}{m}\ln m\right) \\
		&= \frac{1}{2i} \int_1^{\infty} \frac{e^{-\zeta}}{\zeta} d\zeta - \frac{1}{2i} \int_1^{im(z-i)} \frac{e^{-u}}{u} + \frac{O(\frac{1}{m})}{u} du + O\left(\frac{1}{m}\ln m\right) \\
		&= \frac{1}{2i} \int_1^\infty \frac{e^{-\zeta}}{\zeta} d\zeta - \frac{1}{2i} \int_1^{im(z-i)} \frac{e^{-\zeta}}{\zeta} d\zeta + O\left(\frac{1}{m} \ln m\right).
	\end{align*}	
\end{proof}

We note that one can similarly estimate near $-i$.

\section{The local parametrices}\label{sec:local-par}

In this section we shall compute how the local paramterices look like, with paying special attention to those parts that depend on $m$.
As the two cases are very similar, we shall only examine the parametrix $P_1$ around $i$ in detail.
As in (4.12) and (4.23)--(4.24) in \cite{dik} we have
\begin{equation}\label{eq:zeta}
	\zeta = m \ln\frac{z}{i} \qquad (z\in U_1)
\end{equation}
and
\begin{equation}\label{eq:Pi}
	P_1(z) = E(z)\Psi_1(\zeta)F_1(z)^{-\sigma_3}z^{\pm m\sigma_3/2} \qquad (z\in U_1),
\end{equation}
where $\pm = +$ when $|z|<1$, and $\pm = -$ when $|z|>1$.
By equations (4.18)--(4.22) in \cite{dik}, one easily sees that the auxiliary function $F_1(z)$ is constant in $U_1$, and its value is
\begin{equation}\label{eq:F1}
	F_1(z) = F_1 := \sqrt{(\lambda-1)e^{i\pi\beta}} = \sqrt{(\lambda+1)e^{-i\pi\beta}} \qquad (z\in U_1).
\end{equation}
The function $E(z)$ is analytic in a neighbourhood of $U_1$ and is defined in (4.47)--(4.50) in \cite{dik}.
What is important for our considerations is that
\begin{equation}\label{eq:E}
	E(z) =
	\left(\begin{matrix}
		0 & E_{12}(z) \\
		E_{21}(z) & 0 \\
	\end{matrix}\right) 
	= m^{-\beta\sigma_3} i^{\frac{m}{2}\sigma_3} \widetilde{E}(z) 
	= m^{-\beta\sigma_3} i^{\frac{m}{2}\sigma_3}
	\left(\begin{matrix}
		0 & \widetilde{E}_{12}(z) \\
		\widetilde{E}_{21}(z) & 0 \\
	\end{matrix}\right),
\end{equation}
where $\widetilde{E}(z)$ is independent of $m$ and analytic in a neighbourhood of $U_1$.
Furthermore, 
\begin{equation*}
	E_{12}(z) = i^{\frac{m}{2}} \calD(z) \zeta^{-\beta}F_1^{-1}e^{i\pi\beta},\;\;\;
	\widetilde{E}_{12}(z) = \calD(z)\left(\ln\frac{z}{i}\right)^{-\beta}F_1^{-1}e^{i\pi\beta} \qquad(|z|<1)
\end{equation*}
and
\begin{equation*}
	E_{21}(z) = -i^{-\frac{m}{2}}\calD(z)^{-1} \zeta^{\beta}F_1e^{-2i\pi\beta},\;\;\;
	\widetilde{E}_{21}(z) = -\calD(z)^{-1}\left(\ln\frac{z}{i}\right)^{\beta}F_1e^{-2i\pi\beta} \qquad(|z|<1).
\end{equation*}
The function $\Psi_1(\zeta)$ is an auxiliary function which is the main ingredient in constructing the local parametrix in \cite{dik}, and which is given explicitly in terms of the confluent hypergeometric function $\psi(a,c;z)$.
We recall the details now. 
\begin{figure}[h]
	\includegraphics[scale=0.25]{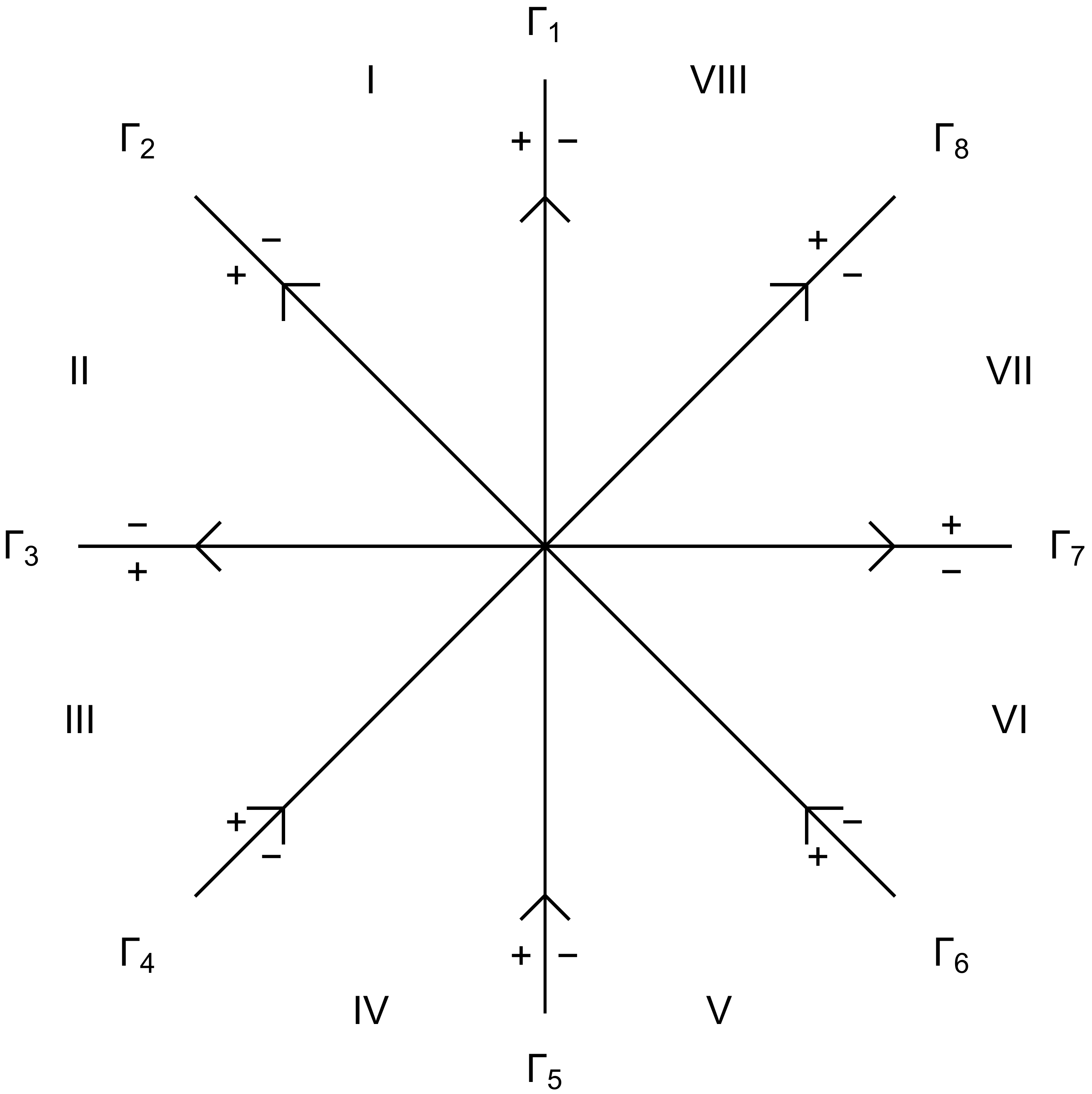}
	\caption{The contour for the local parametrix $\Psi_1$.}\label{fig:Psi-cont}
\end{figure}
Let the contours $\Gamma_1, \dots, \Gamma_8$ be defined as in Figure \ref{fig:Psi-cont}. 
In particular, each of them is a half line starting or ending at $0$.
Furthermore, $\Gamma_k\cup\Gamma_{k+4}$ is a line ($k=1,2,3,4$), and when $k=1, 3$, these unions are the imaginary and real axes, respectively.
These contours divide the complex plane into $8$ sectors, denoted by $I,II,\dots,VIII$ as shown in Figure \ref{fig:Psi-cont}.
The function $\Psi_1(\zeta)$ is analytic in $\C\setminus\cup_{k=1}^8\Gamma_k$, and is uniquely defined by 
\begin{equation*}
	\Psi_1(\zeta) = \Psi_1^{(I)}(\zeta) 
	= \left(\begin{matrix}
		\psi(\beta, 1; \zeta) e^{i\pi 2\beta} e^{-\zeta/2} & -\psi(1-\beta, 1, e^{-i\pi}\zeta) e^{i\pi\beta} e^{\zeta/2}\frac{\Gamma(1-\beta)}{\Gamma(\beta)} \\
		-\psi(1+\beta, 1; \zeta) e^{i\pi\beta} e^{-\zeta/2} \frac{\Gamma(1+\beta)}{\Gamma(-\beta)} & \psi(-\beta, 1, e^{-i\pi}\zeta) e^{\zeta/2} \\
	\end{matrix}\right) 
	\quad (\zeta\in I),
\end{equation*}
and the following jump condition
\begin{equation*}
	\Psi_{1,+}(\zeta) = \Psi_{1,-}(\zeta) J_k(\zeta) 
	\qquad (\zeta\in\Gamma_k),
\end{equation*}
where the jump matrices $J_k$ are constant and are given by
\begin{align*}
	J_1 = \left(\begin{matrix}
		0 & e^{-i\pi\beta} \\
		-e^{i\pi\beta} & 0
	\end{matrix}\right), \;\;\;
	J_2 = J_8 = \left(\begin{matrix}
		1 & 0 \\
		e^{i\pi\beta} & 1
	\end{matrix}\right), \;\;\;
	J_3 = J_7 = \left(\begin{matrix}
		1 & 0 \\
		0 & 1
	\end{matrix}\right), \\
	J_4 = J_6 = \left(\begin{matrix}
		1 & 0 \\
		e^{-i\pi\beta} & 1
	\end{matrix}\right), \;\;\;
	J_5 = \left(\begin{matrix}
		0 & e^{i\pi\beta} \\
		-e^{-i\pi\beta} & 0
	\end{matrix}\right),
\end{align*}
see (4.25)--(4.29) and (4.32) in \cite{dik}.
Note that the functions $\psi(a,c,\zeta)$ and $\psi(a,c,e^{-i\pi}\zeta)$ are defined on the universal covering of the punctured plane $\zeta\in\C\setminus\{0\}$, and that $\Psi_1^{(I)}(\zeta)$ is the analytic continuation of $\Psi_1|_I$ to $0 < \arg\zeta < 2\pi$.

\section{The contribution from $R-I$}\label{sec:before-final}
Recall that the the integration in \eqref{eq:mut-inf} will be taken over the circle $(1+\frac{\varepsilon}{2})C$, hence from now on we only consider the case when $|\lambda| > 1$, which implies $|\Re\beta|<1/4$.
In this section our aim is to show that for $|\lambda| > 1$ the integral 
\begin{equation}\label{eq:int-1}
	\oint_{|z|= 2} \left(R_{11}(z)-1\right) \left(\frac{z-i}{z+i} \right)^\beta \frac{\widetilde{f_m}(z)}{z} dz 
\end{equation}
introduced in \eqref{eq:contributions} converges to 0, and thus the contribution to our inner product coming from $R-I$ is, roughly speaking, negligible. 
Of course, the contour of integration can be deformed to the outer boundary of the unbounded component of $\C\setminus\Gamma_R$. 
Since the integrand is analytic outside $\Gamma_R\cup [-i,i]$, the integrals over the other contours shown on Figure \ref{fig:cont-def} vanish.
Therefore, by a straightforward calculation we obtain the following expression for \eqref{eq:int-1} where $\Delta(z)+I$ is the jump in the $R$--RHP:
\begin{align}
	&- \int_{\Gamma_R} \left(R_{+}(z)-R_{-}(z)\right)_{11} \left(\frac{z-i}{z+i} \right)^\beta \frac{\widetilde{f_m}(z)}{z} dz
	+ \int_{\gamma_m} \left(R_{11}(z)-1\right) \left(\frac{z-i}{z+i} \right)^\beta \frac{\widetilde{f_m}(z)}{z} dz \\
	&= - \int_{\Gamma_R} \left(R_{-}(z)\Delta(z)\right)_{11} \left(\frac{z-i}{z+i} \right)^\beta \frac{\widetilde{f_m}(z)}{z} dz
	+ \int_{\gamma_m} \left(R_{11}(z)-1\right) \left(\frac{z-i}{z+i} \right)^\beta \frac{\widetilde{f_m}(z)}{z} dz, \label{eq:int-2}
\end{align}
where $\gamma_m$ is the union of two circles of radius $1/m$, four linesegments and two half-cirlces of radius $1/2$. More precisely, $\gamma_m = C^i_m \cup [(1-1/m)i,i/2]_+ \cup [-i/2,-(1-1/m)i]_+ \cup C^{-i}_m \cup [-(1-1/m)i,i/2]_- \cup [i/2,(1-1/m)i]_- \cup \{z\colon |z|=1/2\}$ oriented in the positive direction where $C^{\pm i}_m$ is the circle around $\pm i$ with radius $\frac{1}{m}$, and the line segment $[-i,i]$ is oriented upwards, hence its $-$/$+$ side is its right/left side.
\begin{figure}[h]
	\includegraphics[scale=0.2]{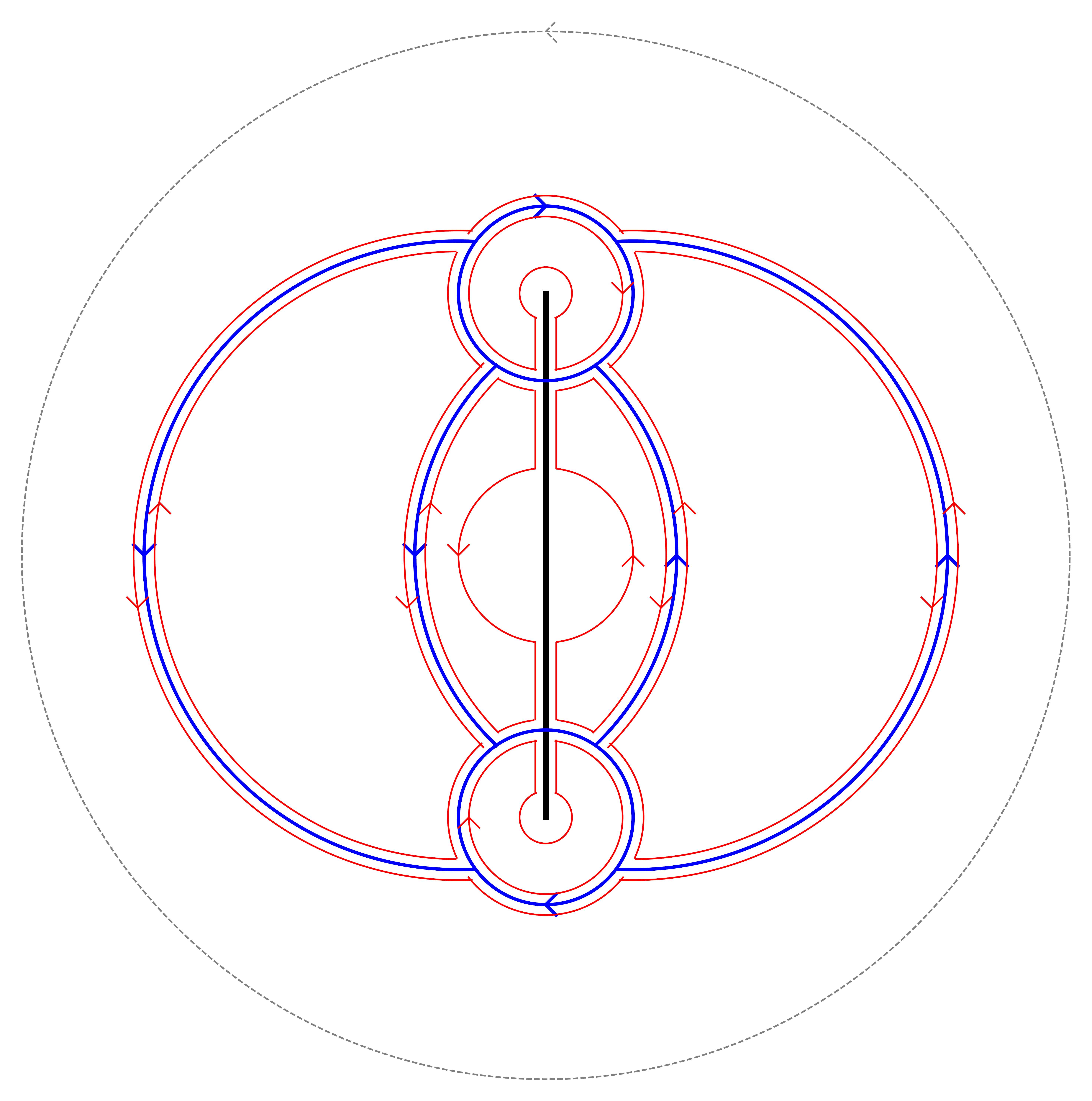}
	\caption{The contour deformation of the dashed circle into the outer red contour. The contour integrals along the other red contours vanish.}\label{fig:cont-def}
\end{figure}

We shall examine the two integrals in \eqref{eq:int-2} separately, starting with the second one.

\subsection{The integral over $\gamma_m$}
\begin{proposition}\label{prop:gamma_m-est}
	We have
	\begin{equation*}
		\int_{\gamma_m} \left(R_{11}(z)-1\right) \left(\frac{z-i}{z+i} \right)^\beta \frac{\widetilde{f_m}(z)}{z} dz = O\left(m^{-1/4}\right) \qquad \text{as}\;\;m\to\infty,
	\end{equation*}
	which is uniform in $\lambda$ on compact subsets of $|\lambda|>1$.
\end{proposition}

\begin{proof}
We know from the standard analysis in the steepest descent method that $\Delta(z) = O(m^{2|\Re\beta|-1})$ and hence $R(z)-I = \frac{1}{2\pi i} \int_{\Gamma_R} \frac{R_-(s)\Delta(s)}{s-z} ds = O(m^{2|\Re\beta|-1})$ as $m\to\infty$, which is uniform in $\lambda$ on compact subsets of $|\lambda|>1$, and in $z$ on $\C\setminus\Gamma_R$ (see e.g. \cite{dik}, and note that there is some flexibility in choosing the parameters for $\Gamma_R$, hence the Cauchy integral does not blow up as $z$ gets closer to $\Gamma_R$).
Also, elementary observations show that $\left(\frac{z-i}{z+i} \right)^\beta = O(m^{|\Re\beta|})$ as $m\to\infty$, which is uniform in $\lambda$ on compact subsets of $|\lambda|>1$, and in $z$ on $\gamma_m$.
Therefore, combining the above estimates with Lemmas \ref{lem:fm-est1} and \ref{lem:fm-est2}, we conclude
\begin{align*}
	\int_{\gamma_m} \left(R_{11}(z)-1\right) \left(\frac{z-i}{z+i} \right)^\beta \frac{\widetilde{f_m}(z)}{z} dz
	& = \int_{\gamma_m} O\left(m^{2|\Re\beta|-1}\right)O\left(m^{|\Re\beta|}\right)O\left(1\right) dz \\
	& = O\left(m^{3|\Re\beta|-1}\right) = O\left(m^{-1/4}\right) \qquad \text{as}\;\; m\to\infty.
\end{align*}
\end{proof}

From now on, we estimate the integral over $\Gamma_R$ from \eqref{eq:int-2}, which we split into two parts.

\subsection{The integrals over $\Sigma_j^{''\text{out}}$ and $\Sigma_j^{\text{out}}$}
First we deal with the integrals over the lenses.

\begin{proposition}\label{prop:lense-est}
	We have
	\begin{equation*}
	\int_{\Sigma_1^{''\text{out}}\cup\Sigma_2^{''\text{out}}\cup\Sigma_1^{\text{out}}\cup\Sigma_2^{\text{out}}} \left(R_{-}(z)\Delta(z)\right)_{11} \left(\frac{z-i}{z+i} \right)^\beta \frac{\widetilde{f_m}(z)}{z} dz = O\left(m^{-1/2}\right),
	\end{equation*}
	which is uniform in $\lambda$ on compact subsets of $|\lambda|>1$.
\end{proposition}

\begin{proof}
Notice that for $z\in \Sigma_j^{''\text{out}}$ we have
\begin{equation}\label{eq:0}
	\Delta(z) = \calD(z)^{\sigma_3}
	\left(\begin{matrix}
		0 & 1 \\
		-1 & 0
	\end{matrix}\right)
	\left(\begin{matrix}
		0 & 0 \\
		\phi(z)^{-1} z^m & 0
	\end{matrix}\right)
	\left(\begin{matrix}
		0 & 1 \\
		-1 & 0
	\end{matrix}\right)
	\calD(z)^{-\sigma_3}
	=
	\left(\begin{matrix}
		0 & * \\
		0 & 0
	\end{matrix}\right).
\end{equation}
Thus $\left(R_{-}(z)\Delta(z)\right)_{11} = 0$, and we conclude that
\begin{equation}\label{eq:inner-lense-est}
	\int_{\Sigma_j^{''\text{out}}} \left(R_{-}(z)\Delta(z)\right)_{11} \left(\frac{z-i}{z+i} \right)^\beta \frac{\widetilde{f_m}(z)}{z} dz = 0 \qquad (j=1,2).
\end{equation}

Next, for $z\in \Sigma_j^{\text{out}}$ we have
\begin{equation*}
	\Delta(z) = \calD(z)^{\sigma_3}
	\left(\begin{matrix}
		0 & 0 \\
		\phi(z)^{-1} z^{-m} & 0
	\end{matrix}\right)
	\calD(z)^{-\sigma_3}
	= \phi(z)^{-1} z^{-m} \left(\frac{z-i}{z+i} \right)^{-2\beta}
	\left(\begin{matrix}
		0 & 0 \\
		1 & 0
	\end{matrix}\right).
\end{equation*}
Note that
\begin{equation*}
	\left|\frac{\widetilde{f_m}(z)}{z^{m+1}}\right| \leq \left|\frac{{f_m}(z)}{z^{m+1}}\right| + \frac{1}{2} \left|\frac{\ln\frac{z-i}{z+i}}{z^{m+1}}\right| \leq \sum_{\ell=1}^m |z|^{\ell-m-1} + O(1) 
	= \frac{1-|z|^{-m}}{|z|-1} + O(1) = O(1)
\end{equation*}
as $m\to\infty$, uniformly in $z$ on $\Sigma_1^{\text{out}}\cup\Sigma_2^{\text{out}}$.
Therefore,
\begin{align}
	\int_{\Sigma_j^{\text{out}}} \left(R_{-}(z)\Delta(z)\right)_{11} \left(\frac{z-i}{z+i} \right)^\beta \frac{\widetilde{f_m}(z)}{z} dz 
	= \int_{\Sigma_j^{\text{out}}} \phi(z)^{-1} \left(\frac{z-i}{z+i} \right)^{-\beta} R_{12,-}(z) \frac{\widetilde{f_m}(z)}{z^{m+1}} dz \nonumber\\
	= \int_{\Sigma_j^{\text{out}}} O(1) O\left(m^{2|\Re\beta|-1}\right) O(1) dz = O\left(m^{-1/2}\right) \label{eq:outer-lense-est}
\end{align}
as $m\to\infty$ ($j=1,2$), which is uniform in $\lambda$ on compact subsets of $|\lambda|>1$.
\end{proof}

Finally, we estimate the integrals over the circles.

\subsection{The integrals over the circles $\partial U_1$ and $\partial U_2$}
We shall only examine the integral over $\partial U_1$ and note that the case of $\partial U_2$ is very similar.
We will deform the contour $\partial U_1$ inside the disk $U_1$.
Recall the jump condition \eqref{eq:R-jump-3}, and that the jump there $P_1(z) N(z)^{-1} = \Delta(z) + I$ is analytic only in a neighbourhood of $U_1\setminus\left(\Sigma_1\cup\Sigma_2\cup\Sigma_1^{''}\cup\Sigma_2^{''}\right)$.
The disk $U_1$ is cut into five components by $[-i,i]\cup\Sigma_1\cup\Sigma_2\cup\Sigma_1^{''}\cup\Sigma_2^{''}$, on all of which the integrand is analytic and continuous up to the boundaries, except maybe at $i$.
Therefore, we can deform the five arcs in the way shown in Figure \ref{fig:inside-U1}, and obtain the following:
\begin{align}
	\int_{\partial U_1} \left(R_{-}(z)\Delta(z)\right)_{11} & \left(\frac{z-i}{z+i} \right)^\beta \frac{\widetilde{f_m}(z)}{z} dz \nonumber \\
	&= -\int_{\gamma_m\cap U_1} \left(R(z)\Delta(z)\right)_{11} \left(\frac{z-i}{z+i} \right)^\beta \frac{\widetilde{f_m}(z)}{z} dz \nonumber \\
	&+ \int_{\Sigma_1^{i,m}\cup\Sigma_2^{i,m}\cup\Sigma_1^{''i,m}\cup\Sigma_2^{''i,m}} \left(R(z)\left[\Delta_+(z)-\Delta_-(z)\right]\right)_{11} \left(\frac{z-i}{z+i} \right)^\beta \frac{\widetilde{f_m}(z)}{z} dz \label{eq:U1-int},
\end{align}
where $\gamma_m$ was defined just after \eqref{eq:int-2}, and $\Sigma_j^{i,m} = \Sigma_j\cap U_1\cap \{z\colon |z-i|>1/m\}$, $\Sigma_j^{''i,m} = \Sigma_j^{''}\cap U_1\cap \{z\colon |z-i|>1/m\}$.
\begin{figure}[h]
	\includegraphics[scale=0.25]{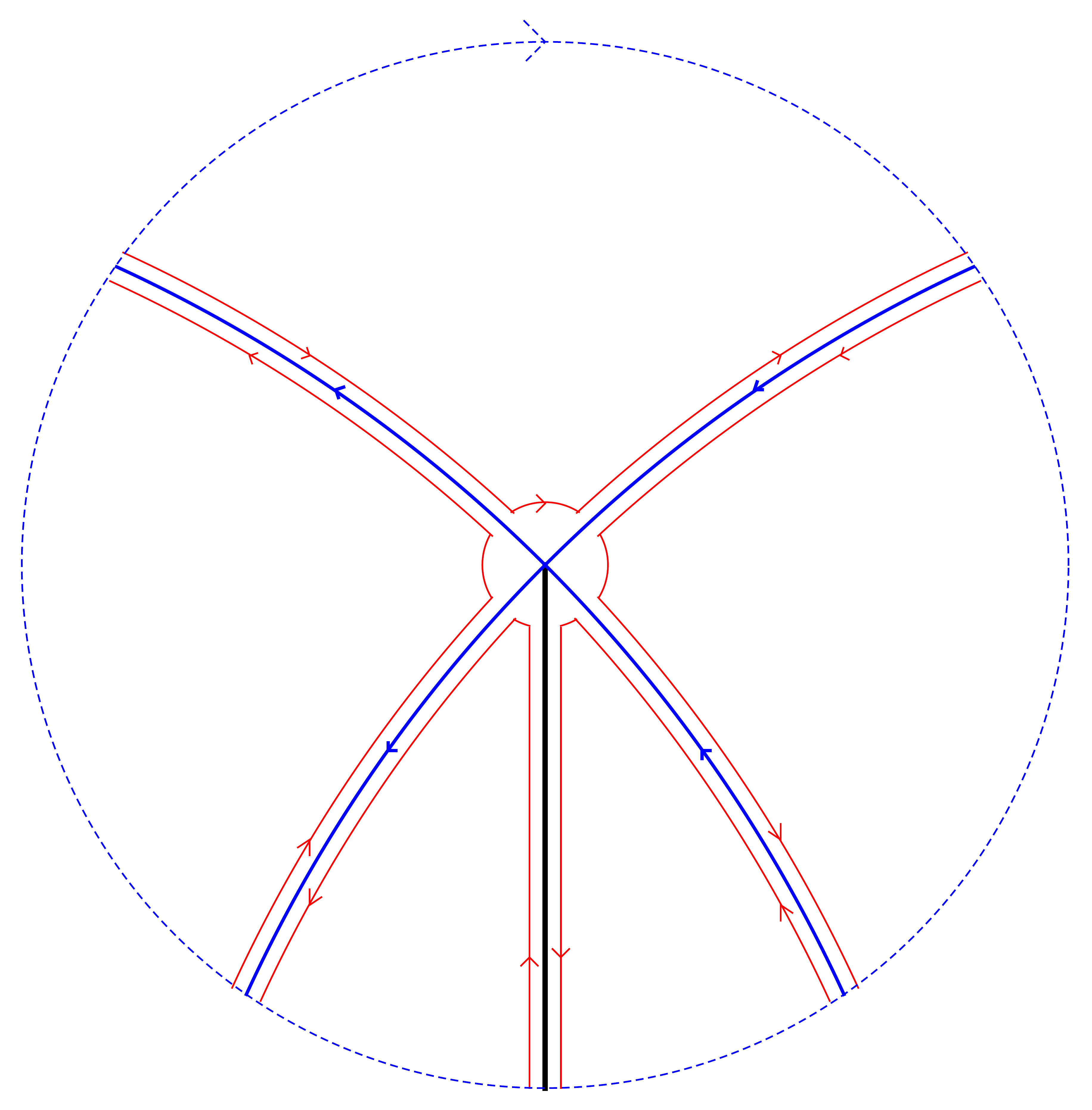}
	\caption{The contour deformation of $\partial U_1$.}\label{fig:inside-U1}
\end{figure}

First, we handle the integral over $\gamma_m\cap U_1$.

\begin{proposition}
	We have 
	\begin{equation}\label{eq:gamma_m-U_i-est}
		\int_{\gamma_m\cap U_1} \left(R(z)\Delta(z)\right)_{11} \left(\frac{z-i}{z+i} \right)^\beta \frac{\widetilde{f_m}(z)}{z} dz = O\left(m^{-1/4}\right) \qquad \text{as} \;\; m\to\infty,
	\end{equation}
	uniformly in $\lambda$ on compact subsets of $|\lambda|>1$. Moreover, the same estimation holds for $U_2$.
\end{proposition}

\begin{proof}
	We have
	\begin{equation*}
		N(z)^{-1} 
		= \left\{\begin{matrix}
			\left(\frac{z-i}{z+i} \right)^{-\beta\sigma_3} & |z|>1, \\
			\left(\begin{matrix}
				0 & -1 \\
				1 & 0
			\end{matrix}\right) \phi(z)^{-\sigma_3} \left(\frac{z-i}{z+i} \right)^{-\beta\sigma_3} & |z|<1,
		\end{matrix}\right.
	\end{equation*}
	hence we obtain $N(z)^{-1} = O(m^{|\Re\beta|})$ as $m\to\infty$, uniformly in $z$ on $C_i^m$ and in $\lambda$ on compact subets of $|\lambda|>1$.
	Using \eqref{eq:zeta}--\eqref{eq:E} we also obtain the following:
	\begin{equation*}
		P_1(z) = m^{-\beta\sigma_3} O(1) \Psi_1\left( m \ln\frac{z}{i} \right) O(1) = m^{-\beta\sigma_3} O(1) = O(m^{|\Re\beta|}) \qquad \left(|z-i|=\frac{1}{m}\right)
	\end{equation*}
	as $m\to\infty$, uniformly in $z$ and in $\lambda$ on compact subets of $|\lambda|>1$.
	Hence $\Delta(z) = P_1(z) N(z)^{-1} - I = O(m^{2|\Re\beta|})$, and therefore
	\begin{align*}
		\int_{C_i^m} \left(R(z)\Delta(z)\right)_{11} \left(\frac{z-i}{z+i} \right)^\beta \frac{\widetilde{f_m}(z)}{z} dz = \int_{C_i^m} O(m^{2|\Re\beta|}) O(m^{|\Re\beta|}) dz = O\left(m^{-1/4}\right)
	\end{align*}
	as $m\to\infty$, uniformly in $\lambda$ on compact subsets of $|\lambda|>1$.
	
Next we show that, by the formulae in Section \ref{sec:local-par}, we have the following for $z\in II\cup III$:
	\begin{align*}
		P_1(z)N(z)^{-1}-I & = E(z)\Psi_1(\zeta)F_1^{-\sigma_3}z^{\frac{m}{2}\sigma_3}
		\left(\begin{matrix}
			0 & -1 \\
			1 & 0
		\end{matrix}\right) \calD(z)^{-\sigma_3} - I \\
		& = \left(\begin{matrix}
			0 & E_{12}(z) \\
			E_{21}(z) & 0
		\end{matrix}\right)
		\Psi_1^{(I)}(\zeta)
		\left(\begin{matrix}
		1 & 0 \\
		e^{i\pi\beta} & 1
		\end{matrix}\right)
		F_1^{-\sigma_3}z^{\frac{m}{2}\sigma_3}
		\left(\begin{matrix}
			0 & -1 \\
			1 & 0
		\end{matrix}\right) \calD(z)^{-\sigma_3} - I \\
		& = F_1 z^{-\frac{m}{2}}
		\left(\begin{matrix}
			* & E_{12}(z)\Psi_{1,22}^{(I)}(\zeta) \\
			* & E_{21}(z)\Psi_{1,12}^{(I)}(\zeta)
		\end{matrix}\right)
		\left(\begin{matrix}
			0 & -1 \\
			1 & 0
		\end{matrix}\right) \calD(z)^{-\sigma_3} - I \\
		& = F_1 z^{-\frac{m}{2}} \calD(z)^{-1} 
		\left(\begin{matrix}
			E_{12}(z)\Psi_{1,22}^{(I)}(\zeta) & * \\
			E_{21}(z)\Psi_{1,12}^{(I)}(\zeta) & *
		\end{matrix}\right)
		- I \\
		& = F_1 z^{-\frac{m}{2}} \calD(z)^{-1} 
		\left(\begin{matrix}
			i^{\frac{m}{2}} \calD(z)\zeta^{-\beta}F_1^{-1}e^{i\pi\beta} \Psi_{1,22}^{(I)}(\zeta) & * \\
			- i^{-\frac{m}{2}} \calD(z)^{-1}\zeta^{\beta}F_1e^{-2i\pi\beta} \Psi_{1,12}^{(I)}(\zeta) & *
		\end{matrix}\right)
		- I \\
		& = 
		\left(\begin{matrix}
			e^{i\pi\beta} e^{-\zeta/2} \zeta^{-\beta}\Psi_{1,22}^{(I)}(\zeta) - 1 & * \\
			-i^{-m} \calD(z)^{-2} F_1^2 e^{-2i\pi\beta} e^{-\zeta/2} \zeta^{\beta}\Psi_{1,12}^{(I)}(\zeta) & *
		\end{matrix}\right) \\
		& = 
		\left(\begin{matrix}
			(e^{-i\pi}\zeta)^{-\beta} \psi(-\beta, 1, e^{-i\pi}\zeta) - 1 & * \\
			i^{-m} \calD(z)^{-2} (e^{-i\pi}\zeta)^{\beta} \psi(1-\beta, 1, e^{-i\pi}\zeta) \frac{\Gamma(1-\beta)}{\Gamma(\beta)} F_1^2  & *
		\end{matrix}\right).
	\end{align*}
In the above we have concentrated on the first 
column only. Denote the radius of $U_1$ by $\varepsilon$. Then,
\begin{align}
		\int_{(\gamma_m\cap U_1)\setminus C_i^m} &\left(R(z)\Delta(z)\right)_{11} \left(\frac{z-i}{z+i} \right)^\beta \frac{\widetilde{f_m}(z)}{z} dz \label{eq:vertical-1} \\
		& = \int_{(1-\varepsilon)i}^{(1-1/m)i} \left(R(z) \left(P_1(z)N(z)^{-1}-I\right)\right)_{11} O\left(m^{|\Re\beta|}\right) \frac{dz}{z} \nonumber \\
		& = \int_{(1-\varepsilon)i}^{(1-1/m)i} R_{11}(z) \left((e^{-i\pi}\zeta)^{-\beta} \psi(-\beta, 1, e^{-i\pi}\zeta) - 1\right) O\left(m^{|\Re\beta|}\right) \frac{dz}{z} \nonumber \\
		& \;\;\;\; + \int_{(1-\varepsilon)i}^{(1-1/m)i} R_{12}(z) \calD(z)^{-2} (e^{-i\pi}\zeta)^{\beta} \psi(1-\beta, 1, e^{-i\pi}\zeta) O\left(m^{|\Re\beta|}\right) \frac{dz}{z}. \nonumber
\end{align}
Finally, by substituting $\xi = e^{-i\pi}\zeta$, $\frac{dz}{z} = -\frac{d\xi}{m}$, and using the large $\xi$ asymptotics of the confluent hypergeometric function and the usual estimates for $R$ and $\calD$, we obtain
\begin{align}
		& = \int^{-m\ln(1-\varepsilon)}_{-m\ln(1-1/m)} \left(\xi^{-\beta} \psi(-\beta, 1, \xi) - 1\right) O\left(m^{|\Re\beta|}\right) \frac{d\xi}{m} + \int^{-m\ln(1-\varepsilon)}_{-m\ln(1-1/m)} \xi^{\beta} \psi(1-\beta, 1, \xi) O\left(m^{5|\Re\beta|-1}\right) \frac{d\xi}{m} \nonumber \\
		& = O\left(m^{|\Re\beta|-1}\right) \int^{-m\ln(1-\varepsilon)}_{1} \xi^{-1} d\xi + O\left(m^{5|\Re\beta|-2}\right) \int^{-m\ln(1-\varepsilon)}_{1} \xi^{2\Re\beta-1} d\xi \nonumber \\
		& = O\left(m^{|\Re\beta|-1}\ln m\right) + O\left(m^{5|\Re\beta|-2}\right) \int^{-m\ln(1-\varepsilon)}_{1} \xi^{-1/2} d\xi 
		= O\left(m^{-1/4}\right) \qquad \text{as} \;\; m\to\infty. \nonumber
\end{align}
\end{proof}

As last step,
we consider the integrals over $\Sigma_j^{i,m}$ and $\Sigma_j^{''i,m}$ $(j=1,2)$.

\begin{proposition}
We have
	\begin{equation}\label{eq:Sigma-U_i-est}
		\int_{\Sigma_1^{i,m}\cup\Sigma_2^{i,m}\cup\Sigma_1^{''i,m}\cup\Sigma_2^{''i,m}} \left(R(z)\left[\Delta_+(z)-\Delta_-(z)\right]\right)_{11} \left(\frac{z-i}{z+i} \right)^\beta \widetilde{f_m}(z) \frac{dz}{z} = O(m^{-1/4}).
	\end{equation}
as $m\to\infty$, uniformly in $\lambda$ on compact subsets of $|\lambda|>1$.
\end{proposition}

\begin{proof}
First, we notice that the integral in \eqref{eq:U1-int} over $\Sigma_j^{''i,m}$ $(j=1,2)$ vanishes.
Indeed, we have $\Delta_+(z)-\Delta_-(z) = \left(P_{1,+}(z)-P_{1,-}(z)\right) N(z)^{-1}$ and the jump of $P_1$ is exactly the same as that of $S$. Therefore, similarly as in \eqref{eq:0} we obtain that $\left(R(z)\left[\Delta_+(z)-\Delta_-(z)\right]\right)_{11} = 0$.

In the rest of the proof, we shall only deal with the integral over $\Sigma_1^{i,m}$, and note that the other integral over $\Sigma_2^{i,m}$ can be handled very similarly, since the local parametrix does not jump along $\Gamma_7$.
Note that for $z\in\Sigma_1^{i,m}$ we have
\begin{align*}
	\Delta_+(z)-\Delta_-(z) &= \left(P_{1,+}(z)-P_{1,-}(z)\right) N(z)^{-1} 
	= P_{1,-}(z) \left(\begin{matrix}
		0 & 0 \\
		\phi(z)^{-1}z^{-m} & 0
	\end{matrix}\right) \calD(z)^{-\sigma_3} \\
	&= m^{-\beta\sigma_3} i^{\frac{m}{2}\sigma_3}
	\left(\begin{matrix}
		0 & \widetilde{E}_{12}(z) \\
		\widetilde{E}_{21}(z) & 0 \\
	\end{matrix}\right)\Psi_{1,-}(\zeta)F_1^{-\sigma_3}z^{-m\sigma_3/2} \left(\begin{matrix}
		0 & 0 \\
		\phi(z)^{-1}z^{-m} & 0
	\end{matrix}\right) \calD(z)^{-\sigma_3} \\
	&= m^{-\beta\sigma_3} i^{\frac{m}{2}\sigma_3}
	\left(\begin{matrix}
		0 & \widetilde{E}_{12}(z) \\
		\widetilde{E}_{21}(z) & 0 \\
	\end{matrix}\right)\Psi_{1,-}(\zeta) \left(\begin{matrix}
		0 & 0 \\
		1 & 0
	\end{matrix}\right) \phi(z)^{-\sigma_3} F_1^{\sigma_3}z^{-m\sigma_3/2} \calD(z)^{-\sigma_3}\\
	&= m^{-\beta\sigma_3} i^{\frac{m}{2}\sigma_3}
	\left(\begin{matrix}
		0 & \widetilde{E}_{12}(z) \\
		\widetilde{E}_{21}(z) & 0 \\
	\end{matrix}\right)
	\left(\begin{matrix}
		\Psi_{1-,12}(\zeta) & 0 \\
		\Psi_{1-,22}(\zeta) & 0
	\end{matrix}\right) \phi(z)^{-1} z^{-m/2} \calD(z)^{-1} F_1 \\
	&= m^{-\beta\sigma_3} i^{\frac{m}{2}\sigma_3}
	\left(\begin{matrix}
		\widetilde{E}_{12}(z) \Psi_{1-,22}(\zeta) & 0 \\
		\widetilde{E}_{21}(z) \Psi_{1-,12}(\zeta) & 0 \\
	\end{matrix}\right)
	\phi(z)^{-1} z^{-m/2} \calD(z)^{-1} F_1 \\
	&= \left(\begin{matrix}
		m^{-\beta} i^{\frac{m}{2}} \widetilde{E}_{12}(z) \Psi_{1-,22}(\zeta) & 0 \\
		m^{\beta} i^{-\frac{m}{2}} \widetilde{E}_{21}(z) \Psi_{1-,12}(\zeta) & 0 \\
	\end{matrix}\right)
	\phi(z)^{-1} z^{-m/2} \left(\frac{z-i}{z+i} \right)^{-\beta} F_1.
\end{align*}
Thus, the integral over $\Sigma_1^{i,m}$ has the following form:
\begin{align}
	\int_{\Sigma_1^{i,m}} & R_{11}(z)m^{-\beta} i^{\frac{m}{2}} \widetilde{E}_{12}(z) \Psi_{1-,22}(\zeta) \phi(z)^{-1} z^{-m/2} F_1 \widetilde{f_m}(z) \frac{dz}{z} \label{eq:S1im-1}\\
	&+ \int_{\Sigma_1^{i,m}} R_{12}(z)m^{\beta} i^{-\frac{m}{2}} \widetilde{E}_{21}(z) \Psi_{1-,12}(\zeta) \phi(z)^{-1} z^{-m/2} F_1 \widetilde{f_m}(z) \frac{dz}{z}. \label{eq:S1im-2}
\end{align}

Note that we have
\begin{align*}
	\Psi_1(\zeta) = \Psi_1^{(I)}(\zeta) J_1^{-1} J_8^{-1} 
	= \Psi_1^{(I)}(\zeta) 
	\left(\begin{matrix}
		1 & -e^{-i\pi\beta} \\
		e^{i\pi\beta} & 0
	\end{matrix}\right)
	= \left(\begin{matrix}
		* & -e^{-i\pi\beta} \Psi_{1,11}^{(I)}(\zeta) \\
		* & -e^{-i\pi\beta} \Psi_{1,21}^{(I)}(\zeta) 
	\end{matrix}\right) \qquad (\zeta\in VII).
\end{align*}
Therefore, if we use the estimation $\left|\frac{f_m(z)}{z^{m}}\right| \leq \sum_{\ell=1}^m \frac{|z|^{\ell-m}}{\ell} \leq \sum_{\ell=1}^m \frac{1}{\ell}$ $(|z|>1)$, the integral \eqref{eq:S1im-1} becomes
\begin{align*}
	&\int_{\Sigma_1^{i,m}} O\left(m^{-\Re\beta}\right) \Psi_{1-,22}(\zeta) z^{-m/2} \widetilde{f_m}(z) \frac{dz}{z} 
	= \int_{\Sigma_1^{i,m}} O\left(m^{-\Re\beta}\right) \Psi_{1,21}^{(I)}(\zeta) z^{-m/2} \left(f_m(z) - \frac{1}{2i}\ln\frac{z-i}{z+i}\right) \frac{dz}{z} \\
	&= \int_{\Sigma_1^{i,m}} O\left(m^{-\Re\beta}\right) \psi(1+\beta, 1; \zeta) \frac{f_m(z) - \frac{1}{2i}\ln\frac{z-i}{z+i}}{z^{m}} \frac{dz}{z}
	= \int_{\Sigma_1^{i,m}} O\left(m^{-\Re\beta}\ln m\right) \psi(1+\beta, 1; \zeta) \frac{dz}{z}, \\
	&= \int_{\Gamma_8\cap\{\zeta\colon 1+O(1/m) \leq |\zeta|\leq Cm\}} O\left(m^{-\Re\beta}\ln m\right) \zeta^{-1-\beta} \frac{d\zeta}{m} = O\left(m^{-\Re\beta}m^{-1}(m^{1/4}+1)\ln m\right) = O(m^{-1/2}),
\end{align*}
where $C>0$ is a constant, $\frac{dz}{z} = \frac{d\zeta}{m}$, and we used the large $\zeta$ asymptotics of $\psi(1+\beta, 1; \zeta)$.
The other integral \eqref{eq:S1im-2} can be estimated somewhat similarly as follows:
\begin{align*}
	&\int_{\Sigma_1^{i,m}} O\left(m^{-1+2|\Re\beta|}\right) m^{\beta} \Psi_{1,12}(\zeta) z^{-m/2} \widetilde{f_m}(z) \frac{dz}{z}
	= \int_{\Sigma_1^{i,m}} O\left(m^{-1+2|\Re\beta|}\right) m^{\beta} \Psi_{1,11}^{(I)}(\zeta) z^{-m/2} \widetilde{f_m}(z) \frac{dz}{z} \\
	&= \int_{\Sigma_1^{i,m}} O\left(m^{-1+2|\Re\beta|}\right) m^{\beta} \psi(\beta, 1; \zeta) \frac{\widetilde{f_m}(z)}{z^{m}} \frac{dz}{z} \\
	&= \int_{\Gamma_8\cap\{\zeta\colon 1+O(1/m) \leq |\zeta|\leq Cm\}} O\left(m^{-1+2|\Re\beta|}\ln m\right) \left(\frac{\zeta}{m}\right)^{-\beta} \frac{d\zeta}{m} \\
	&= O\left(m^{-1+2|\Re\beta|}\ln m\right) \int_{0}^C t^{-\Re\beta} dt = O\left(m^{-1+2|\Re\beta|}\ln m\right) = O\left(m^{-1/2}\right).
\end{align*}
\end{proof}

\textcolor{black}{With the above proof we have finished proving Lemma \ref{lem:main}, which we use in the next section.}

\section{Calculating the mutual information}\label{sec:final}
In this section we prove Theorem \ref{thm:main}, that is, we calculate \eqref{eq:mut-inf}.
We start with a lemma.

\begin{lemma}\label{lem:beta-est}
	We have
	\begin{equation*}
		\color{black}\left|\tan\left(\tfrac{\pi}{2}\beta\right)\right| \color{black}< 1 \qquad (\lambda\in\C\setminus[-1,1]).
	\end{equation*}
\end{lemma}

\begin{proof}
	This is a simple geometric observation. Consider the parallelogram on the complex plane with vertices $0$, $e^{i\frac{\pi}{2}\beta}$, $e^{-i\frac{\pi}{2}\beta}$ and $e^{i\frac{\pi}{2}\beta}+e^{-i\frac{\pi}{2}\beta}$. Notice that $|\arg (e^{\pm i\frac{\pi}{2}\beta})|<\pi/4$, hence the angle in the parallelogram at $0$ is less than $\pi/2$. Therefore $|e^{i\frac{\pi}{2}\beta} - e^{-i\frac{\pi}{2}\beta}| < |e^{i\frac{\pi}{2}\beta} + e^{-i\frac{\pi}{2}\beta}|${, which completes the proof}.
\end{proof}

Integration by parts gives the following for all $\varepsilon>0$ and $m,n\in\N$:
\begin{align*}
	&\frac{-1}{2\pi i} \oint_{\Gamma_\varepsilon} e(1+\varepsilon,\lambda)\frac{d}{d\lambda}\ln \widehat{D}(\lambda) \; d\lambda = \frac{1}{2\pi i} \oint_{\Gamma_\varepsilon} \ln \widehat{D}(\lambda) \frac{d}{d\lambda}e(1+\varepsilon,\lambda) \; d\lambda \\
	&= \frac{1}{2\pi i} \oint_{\Gamma_\varepsilon} \ln\left(1 - \left\langle\vecG,\vecg\right\rangle\left\langle\vecGG,\vecgg\right\rangle\right) \frac{d}{d\lambda}e(1+\varepsilon,\lambda) \; d\lambda \\
	& \color{black} = \frac{1}{2\pi i} \oint_{\Gamma_\varepsilon} \ln\left(1 - \left[i\tan\left(\tfrac{\pi}{2}\beta\right) + O(m^{-1/4})\right]\left[i\tan\left(\tfrac{\pi}{2}\beta\right) + O(n^{-1/4})\right]\right) \frac{d}{d\lambda}e(1+\varepsilon,\lambda) \; d\lambda \\
	&\color{black} = \frac{1}{2\pi i} \oint_{\Gamma_\varepsilon} \ln\left(1 + \tan^2\left(\tfrac{\pi}{2}\beta\right) + O(m^{-1/4}+n^{-1/4})\right)\frac{d}{d\lambda}e(1+\varepsilon,\lambda) \; d\lambda.
\end{align*}
Therefore, equation \eqref{eq:mut-inf} becomes
\begin{align}
	& \color{black}\lim_{\varepsilon\searrow 0} \frac{1}{2\pi i} \oint_{\Gamma_\varepsilon} \ln\left(1 + \tan^2\left(\tfrac{\pi}{2}\beta\right)\right) \frac{d}{d\lambda}e(1+\varepsilon,\lambda) \; d\lambda \nonumber \\
	& \color{black}=\lim_{\varepsilon\searrow 0} \frac{-1}{2\pi i} \oint_{\Gamma_\varepsilon} e(1+\varepsilon,\lambda) \frac{d}{d\lambda}\ln\left(1 + \tan^2\left(\tfrac{\pi}{2}\beta\right)\right) \; d\lambda, \label{eq:mut-inf-2}
\end{align}
where we used integration by parts.
Note that \color{black}
\begin{align*}
	&\frac{d}{d\lambda}\ln\left(1 + \tan^2\left(\tfrac{\pi}{2}\beta\right)\right) = \frac{d}{d\beta}\ln\left(1 + \tan^2\left(\tfrac{\pi}{2}\beta\right)\right) \frac{d\beta}{d\lambda} 
	= -i \tan\left(\tfrac{\pi}{2}\beta\right) \frac{1}{1-\lambda^2}.
\end{align*}
\color{black} Hence, equation \eqref{eq:mut-inf-2} becomes \color{black}
\begin{align*}
	\lim_{\varepsilon\searrow 0} \frac{1}{2\pi} \oint_{\Gamma_\varepsilon} e(1+\varepsilon,\lambda) 
	\tan\left(\tfrac{\pi}{2}\beta\right) \frac{1}{1-\lambda^2}
	 \; d\lambda.
\end{align*}
\color{black} Note that by Cauchy's theorem there is a flexibility in choosing $\Gamma_\varepsilon$. 
We observe that for $|\lambda-1|=\frac{\varepsilon}{2}$ we have $e(1+\varepsilon,\lambda) = -\frac{1+\varepsilon+\lambda}{2}\ln\frac{1+\varepsilon+\lambda}{2} - \frac{1+\varepsilon-\lambda}{2}\ln\frac{1+\varepsilon-\lambda}{2} = O(\varepsilon)+O(\varepsilon\ln\varepsilon) = O(\varepsilon\ln\varepsilon)$ as $\varepsilon\searrow 0$.
Therefore, if $C_1^\varepsilon$ denotes the circle around $1$ with radius $\frac{\varepsilon}{2}$, then \color{black}
\begin{align*}
	\lim_{\varepsilon\searrow 0} \frac{1}{2\pi} \oint_{C_1^\varepsilon} e(1+\varepsilon,\lambda) 
	\tan\left(\tfrac{\pi}{2}\beta\right) \frac{1}{1-\lambda^2}
	 \; d\lambda 
	 = \lim_{\varepsilon\searrow 0} \oint_{C_1^\varepsilon} O(\ln\varepsilon) \; d\lambda = 0.
\end{align*}
\color{black}
We similarly get that the integral over $C_{-1}^\varepsilon$ converge to 0.
Therefore, \textcolor{black}{by the Lebesgue dominant convergence theorem,} we conclude that \eqref{eq:mut-inf-2} is equal to the following: \color{black}
\begin{align}
	&\lim_{\varepsilon\searrow 0} \frac{1}{2\pi} \int_{-1+\frac{\varepsilon}{2}}^{1-\frac{\varepsilon}{2}} e(1+\varepsilon,\lambda) 
	\left[ \tan\left(\tfrac{\pi}{2}\beta_+\right) - \tan\left(\tfrac{\pi}{2}\beta_-\right) \right] \frac{1}{1-\lambda^2}
	 \; d\lambda \nonumber\\
	 & =\frac{1}{2\pi} \int_{-1}^{1} e(1,\lambda) 
	\left[ \tan\left(\tfrac{\pi}{2}\beta_+\right) - \tan\left(\tfrac{\pi}{2}\beta_-\right) \right] \frac{1}{1-\lambda^2}
	 \; d\lambda, \label{eq:mut-inf-3}
\end{align}
\color{black} where the interval $[-1,1]$ is oriented from the right to the left, hence its $+/-$ sides are its below/upper sides.
Note that $\lambda \mapsto \beta$ transforms $\C\setminus [-1,1]$ onto $\C\setminus (-\infty,0]$. Also, if $(-\infty,0]$ is oriented from the left to right, then the $+$ side of $[-1,1]$ is mapped onto the $+$ side of $(-\infty,0]$. In particular, 
\begin{equation*}
	\beta_{\pm} = \frac{1}{2\pi i} \left(\ln\left|\tfrac{\lambda+1}{\lambda-1}\right| \pm \pi i\right)
	= \frac{1}{2\pi i} \left(\ln\left(\tfrac{1+\lambda}{1-\lambda}\right) \pm \pi i\right)
	\qquad (\lambda\in(-1,1)),
\end{equation*}
\color{black} and hence a straightforward calculation shows that
\begin{align*}
	\tan\left(\tfrac{\pi}{2}\beta_+\right) - \tan\left(\tfrac{\pi}{2}\beta_-\right) 
	= 2\sqrt{1-\lambda^2} \qquad (\lambda\in(-1,1)).
\end{align*}

Now, plugging in the above into \eqref{eq:mut-inf-3} we obtain
\begin{align*}
	& \frac{-1}{\pi} \int_{-1}^{1} 
	\left(\frac{1+\lambda}{2}\ln\frac{1+\lambda}{2} + \frac{1-\lambda}{2}\ln\frac{1-\lambda}{2}\right)
	\frac{1}{\sqrt{1-\lambda^2}}
	 \; d\lambda = \frac{-1}{\pi} \int_{-1}^{1} 
	\sqrt{\frac{1+\lambda}{1-\lambda}}\ln\frac{1+\lambda}{2}
	 \; d\lambda \\
	 & = \frac{-2}{\pi} \int_{0}^{1} 
	\sqrt{\frac{t}{1-t}}\ln t \; dt 
	= \frac{-8}{\pi} \int_{0}^{1} 
	\frac{v^2}{\sqrt{1-v^2}}\ln v \; dv \\
	& = \frac{8}{\pi} \left( \int_{0}^{1} 
	\sqrt{1-v^2}\ln v \; dv
	- \int_{0}^{1} 
	\frac{\ln v}{\sqrt{1-v^2}} \; dv \right)
	= 2\ln 2 - 1. \\
\end{align*}
Above we performed two substitutions $2t = 1+\lambda$ and $t=v^2$. Since the last two special integrals are well known to be equal to $-\frac{\pi}{8}-\frac{\pi}{4}\ln 2$ and $-\frac{\pi}{2}\ln 2$, respectively,
the proof of Theorem \ref{thm:main} is complete. \color{black}

\section{Conclusions}
Understanding entanglement in bipartite systems is of fundamental importance in quantum information,
but at the same time it is often fraught with technical difficulties.  
One of the major reasons that makes it such a challenging problem is that even
in simple systems the calculations are rather involved and often  it is impossible to 
perform a rigorous analysis. Over the past twenty years a lot of research on bipartite 
entanglement has focused on one-dimensional quantum lattice models, because they are
amenable to a certain degree of mathematical manipulations. Originally, the interest concentrated
on the entanglement entropy of a single interval of contiguous spins with the rest of the chain~\cite{AOPFP04,CaCa04,
FPS09,IJK05,IJK,IMM08,JK04,KM04,KM05,KMN,Kor04,ON02,VLRK03} but more recently
physicists have directed their research on the computation of the entanglement entropy of disjoint 
blocks~\cite{ATC10,AEF14,CCT09,CCT11,Car13,FC10,JK11,MVS11}.

In this article we study a quadratic form of Fermi operators and compute the von Neumann entropy of two
disjoint intervals separated by one lattice site.  The major contribution of this paper 
is a rigorous analysis of the asymptotic limit of the entropy as the size of the intervals tends to infinity. 
More precisely, let $P=P_1 \cup P_2$, where $P_1=\{1,\dotsc,m\}$ and $P_2=\{m+2,\dotsc,m + n+1\}$.  Write $S(\rho^{(m)}_P)$
and $S\left(\rho^{(n)}_P\right)$ for the entropies of the blocks of fermions at $P_1$ and $P_2$, respectively; denote by $S(\rho_P)$ the
entanglement entropy between $P$ and the rest of the chain.   The quantity 
$S(\rho^{(m)}_P)$ was computed in~\cite{JK04}. We prove
that the \emph{mutual entropy} between $P_1$ and $P_2$ 
\begin{equation}
S(\rho^{(m)}_P)+S\left(\rho^{(n)}_P\right)-S(\rho_P) \to 2\ln 2 - 1,
\end{equation}
as $m,n \to\infty$. The proof is based on the Riemann-Hilbert method and involves computing the asymptotics of a Toeplitz determinant as well as 
extracting precise information on the asymptotic behaviour of the inverse of a Toeplitz matrix.  Besides the intrinsic physical
importance of the problem, the mathematical result is of interest in its own right.   The asymptotic analysis of Toeplitz determinants has a long history and is still an area of active research. The ultimate goal would be to compute rigorously the entanglement entropy of two disjoint 
gaps separated by an arbitrary number of lattice sites.  Unfortunately, this is still beyond our present ability.

\end{document}